\documentclass[11pt]{article}
\usepackage{booktabs} %
\usepackage[ruled]{algorithm2e} %

\SetAlFnt{\small}
\SetAlCapFnt{\small}
\SetAlCapNameFnt{\small}
\SetAlCapHSkip{0pt}
\IncMargin{-\parindent}

\usepackage{hyperref}
\usepackage{amsmath}
\usepackage{amsthm}
\usepackage{amsfonts}
\usepackage{amssymb}
\usepackage{bbm}
\usepackage{bm}
\usepackage{enumitem}
\usepackage{float}
\usepackage{mathrsfs}
\usepackage{mathtools}
\usepackage{subcaption}
\usepackage{thmtools}
\usepackage{thm-restate}
\usepackage[normalem]{ulem}

\usepackage[utf8]{inputenc}
\usepackage{xcolor}
\usepackage{fullpage}

\newtheorem{lemma}{Lemma}

\newtheorem{corollary}{Corollary}
\newtheorem*{remark}{Remark}

\definecolor{cobalt}{rgb}{0.0, 0.28, 0.67}

\ifdefined\DRAFT

\newcommand{\gz}[1]{\noindent \textcolor{blue}{{\textbf{GZ:} #1}}}
\newcommand{\hl}[1]{\noindent \textcolor{magenta}{{\textbf{HL:} #1}}}
\newcommand{\gw}[1]{\noindent \textcolor{red}{{\textbf{GW:} #1}}}
\newcommand{\rj}[1]{\noindent \textcolor{purple}{{\textbf{RJ:} #1}}}

\newcommand{\delgz}[1]{{\color{blue} \sout{#1}}}
\newcommand{\delgw}[1]{{\color{red} \sout{#1}}}
\newcommand{\delrj}[1]{{\color{purple} \sout{#1}}}

\newcommand{\todo}[1]{\noindent \textcolor{red}{{\textbf{TODO:} #1}}}

\else
\newcommand{\rj}[1]{}

\newcommand{\delrj}[1]{}

\newcommand{\gw}[1]{}

\newcommand{\delgw}[1]{}
\newcommand{\hl}[1]{}

\newcommand{\delhl}[1]{}

\newcommand{\gz}[1]{}

\newcommand{\delgz}[1]{}

\newcommand{\todo}[1]{}

\newcommand{\ECdelete}[1]{}

\fi

\DeclareMathOperator{\Var}{Var}         %
\DeclareMathOperator{\Cov}{Cov}         %
\DeclareMathOperator*{\argmin}{argmin}  %
\DeclareMathOperator{\Unif}{Unif}
\DeclareMathOperator{\Binom}{Binom}
\DeclareMathOperator{\Poisson}{Poisson}

\newcommand{\E}{\mathbb{E}}
\newcommand{\R}{\mathbb{R}}

\renewcommand{\v}[1]{\ensuremath{\boldsymbol{#1}}}

\newcommand{\MSE}{\ensuremath{\mathsf{MSE}}}
\newcommand{\GTE}{\ensuremath{\mathsf{GTE}}}
\newcommand{\LR}{\ensuremath{\mathsf{LR}}}
\newcommand{\CR}{\ensuremath{\mathsf{CR}}}
\newcommand{\GC}{\ensuremath{\mathsf{GC}}}
\newcommand{\GT}{\ensuremath{\mathsf{GT}}}

\newcommand{\conprob}{\ensuremath{p}}
\newcommand{\tildeconprob}{\ensuremath{\tilde{p}}}
\newcommand{\appprob}{\ensuremath{q}}

\newcommand{\bookprob}{\ensuremath{r}}

\newcommand{\conrate}{\ensuremath{\phi}}
\newcommand{\conratematrix}{\ensuremath{\Phi}}
\newcommand{\tildeconrate}{\ensuremath{\tilde{\phi}}}
\newcommand{\tildeconratematrix}{\ensuremath{\tilde{\Phi}}}
\newcommand{\apprate}{\ensuremath{\psi}}
\newcommand{\appratematrix}{\ensuremath{\Psi}}
\newcommand{\tildeapprate}{\ensuremath{\tilde{\psi}}}
\newcommand{\tildeappratematrix}{\ensuremath{\tilde{\Psi}}}
\newcommand{\bookrate}{\ensuremath{\omega}}
\newcommand{\bookratematrix}{\ensuremath{\Omega}}
\newcommand{\tildebookrate}{\ensuremath{\tilde{\omega}}}
\newcommand{\tildebookratematrix}{\ensuremath{\tilde{\Omega}}}

\newcommand{\ginG}{\gamma\in\Gamma}
\newcommand{\tinT}{\theta\in\Theta}

\newcommand{\bsigma}{\v{\sigma}}
\newcommand{\gtehatlr}{\widehat{\GTE}_{\LR}}
\newcommand{\gtehatcr}{\widehat{\GTE}_{\CR}}

\begin{document}

\title{Interference, Bias, and Variance in Two-Sided Marketplace Experimentation: Guidance for Platforms}

\author{
Hannah Li \\
hannahli@stanford.edu
\\
\\
Ramesh Johari\\
rjohari@stanford.edu
\and 
Geng Zhao \\
gengz@stanford.edu
\\ \\
Gabriel Y. Weintraub \\
gweintra@stanford.edu
}

\maketitle

\begin{abstract}
Two-sided marketplace platforms often run experiments (or A/B tests) to test the effect of an intervention before launching it platform-wide. A typical approach is to randomize individuals into the treatment group, which receives the intervention, and the control group, which does not. The platform then compares the performance in the two groups to estimate the effect if the intervention were launched to everyone.  We focus on two common experiment types, where the platform randomizes individuals either on the supply side or on the demand side. For these experiments, it is well known that the resulting estimates of the treatment effect are typically biased: because individuals in the market compete with each other, individuals in the treatment group affect individuals in the control group (and vice versa), creating interference and leading to a biased estimate. 
       
We develop a simple, tractable market model to study bias and variance in these experiments with interference. We focus on two choices available to the platform: (1) Which side of the platform should it randomize on (supply or demand)? (2) What proportion of individuals should be allocated to treatment?  We find that both choices affect the bias and variance of the resulting estimators, but in different ways. The bias-optimal choice of experiment type depends on the relative amounts of supply and demand in the market, and we discuss how a platform can use market data to select the experiment type. Importantly, we find that in many circumstances, choosing the bias-optimal experiment type has little effect on variance, and in some cases can coincide with the variance-optimal type. On the other hand, we find that the choice of treatment proportion can induce a bias-variance tradeoff, where the bias-minimizing proportion actually increases variance. We discuss how a platform can navigate this tradeoff and best choose the treatment proportion, using a combination of modeling as well as contextual knowledge about the market, the risk of the intervention, and reasonable effect sizes of the intervention.
\end{abstract}

\section{Introduction}
\label{sec:intro}

Two-sided marketplace platforms often run experiments (or A/B tests) to test the effect of an intervention on a subset of the platform before launching it platform-wide. This experimentation approach allows platforms to make data-driven decisions in deciding what features to launch and allows platforms to try out risky, but potentially beneficial, ideas before committing to them \cite{kohavi2020trustworthy}. Our focus is on experiments in two-sided marketplaces, which include markets for freelancing, ridesharing, and lodging, among others. 

When a platform runs an experiment, the goal is to estimate the effect that an intervention would have on a metric of interest if it were launched to the entire platform, compared to the case when the intervention is not introduced to anyone; we call this effect the \textit{global treatment effect} or $\GTE$.
A typical experimental approach is to randomize individuals into the treatment group, which receives the intervention, and the control group, which does not. In the two-sided markets we consider, there are two natural types of experiments that are typically run in practice: one that randomizes on the supply side (which we call listing-side randomization, or $\LR$) and one that randomizes on the demand side (which we call customer-side randomization, or $\CR$). 

In both of these types of experiments, the standard difference-in-means estimators are often biased estimates of the true $\GTE$. Individuals in the market interact and compete with each other, creating interference between units and violating the typical Standard Unit Treatment Value Assumption (SUTVA) that guarantees unbiased estimators. Such interference can lead to biased estimates \cite{ImbensRubin15}.

To see how bias arises due to marketplace competition, consider a $\CR$ experiment where customers are randomized into treatment and control groups. In the experiment, {\em both} treatment and control customers interact with the same supply, and thus, if a treatment customer is able to make a purchase, that mechanically implies a reduction in effective supply for control customers.   %
These interactions lead to interference in a statistical sense, and create a bias in the resulting estimators. A similar argument applies to estimators resulting in $\LR$ experiments, which are biased because {\em both} treatment and control listings interact with the same customers.

The existence of interference on marketplace experiments is well documented (see related work below); indeed, previous studies have shown that the resulting bias can be as large as the $\GTE$ itself \cite{Blake14,holtz2020reducing,Fradkin2015SearchFA}. To address this problem, many platforms have adopted alternative experiment designs such as clustered experiments or switchback experiments, where the platform randomizes on geographical units or intervals of time instead of randomizing individuals \cite{Ugander13,chamandy16,sneider19,bojinov2021design}. These types of experiments can decrease bias but also increase variance. The implementation of these designs is also more complicated than designs that randomize on individuals and, for many platforms, can create significant engineering challenges; improperly implemented, such designs can suffer from bias as well \cite{glynn2020adaptive,Ugander13}. For these reasons, $\CR$ and $\LR$ experiments continue to be popular designs within two-sided platforms. 

The goal of this work is to investigate the bias and variance of these simpler designs, with the aim of providing guidance to platforms on the use of these designs.  We are particularly interested in the impact of two choices: (1) which side of the platform to randomize on (customers or listings); and (2) the proportion of individuals allocated to the treatment group (the treatment allocation).  Our work contributes to the toolkit of techniques available to two-sided platforms to reduce estimation error, and in particular guides platforms without the engineering resources available to implement more complicated designs.  We analyze bias and variance of both $\CR$ and $\LR$ estimators, and we discuss the practical implications for a platform. 

Below we describe our main contributions in more detail.

\noindent\textbf{Market model for study of experimental designs.} 
We develop a simple, tractable market model that captures the relevant competition effects leading to interference (Section \ref{sec:model}). %
The model consists of $N$ listings on the supply side and $M$ customers on the demand side. We allow for heterogeneity on both sides. Customers book (or buy) a listing through a one-shot model that involves three steps. First, each customer forms a consideration set from the set of listings. In this step, the customer includes each listing in its consideration set independently with some probability that depends on both the customer type and listing type. We call this probability the \textit{consideration probability}. Then, each customer applies to a listing in its consideration set at random. Finally, a listing sees the set of applications it received and, if it received at least one application, it accepts an application at random.  

This booking process captures competition among customers and among listings. Because a customer can apply to at most one listing, the listings are in competition with each other for applications. Likewise, because each listing can accept at most one application, customers are in competition with each other for resources.  These competition effects capture the interactions leading to interference in marketplace experiments.  

We describe how to use this model to study experimental designs (Section \ref{sec:experiment_design}).
We focus on interventions that change the consideration set formation process, specifically those that change the consideration probabilities. Such interventions include those that modify a platform’s interface, the amount of information shown about a listing, or the search and recommendation system. Changes in these consideration probabilities propagate in the booking process to also change the probability that a customer will apply to a given listing and the ultimate probability of booking. %
 We study experiments that randomize on the supply side, which we call listing-randomized ($\LR$) designs, and experiments that randomize on the demand side, which we call customer-randomized ($\CR$) designs.

\noindent\textbf{Characterization of bias and variance.} 
The competition that arises between customers and between listings creates interdependencies that complicate the analysis. To make the analysis tractable, we consider a large market regime in which both the number of listings and customers scale to infinity proportionally. Section \ref{sec:limiting_behavior} characterizes the booking behavior in this large market setting.
We use these characterizations to study bias and variance of estimators in our experiments.  Section \ref{sec:bias_and_var_of_lr_cr} derives expressions for the bias and variance of estimators in the $\LR$ and $\CR$ designs, as a function the experiment type, market conditions, the change in choice probability, and the proportion of customers allocated to treatment. %

The bias and variance of an experiment depend on both market conditions, such as the ratio of supply and demand in the market, as well as the decisions that that the platform makes when running an experiment. That is, there are some factors which affect the experiment outcomes that are \textit{beyond the control} of the platform (at least, in the short term) and there are other factors that the platform \textit{can} control.   In the remainder of the paper, we then focus on the factors that the platform can control, namely the experiment type and the proportion allocated to treatment.

\noindent\textbf{Optimizing choice of experiment type.}  For some (but not all) interventions, the platform can choose whether to run a $\CR$ or $\LR$ experiment. Consider an intervention that provides further information on a listing, e.g., lengthening its description. This intervention could be tested through either a $\CR$ or $\LR$ experiment. In a $\CR$ experiment, customers would be randomized into treatment and control;  treatment customers would see listings with lengthier descriptions and control customers would see listings with the original descriptions. In an $\LR$ experiment, listings would be randomized into treatment and control; all customers would see treatment listings with lengthier descriptions and control listings with the original descriptions.

Section \ref{sec:optimizing_exp_type} shows that, among $\CR$ and $\LR$ designs, the bias-optimal experiment type depends on the relative amounts of supply and demand in the market. We provide a theorem showing that the difference-in-means $\CR$ estimator becomes unbiased as market demand becomes small and that the difference-in-means $\LR$ estimator becomes unbiased as market demand grows large. Further, we find through calibrated simulations that choosing the bias-optimal design between $\CR$ and $\LR$ has little effect on variance, and in some cases can coincide with the variance-optimal design. Hence, in short, if the choice of experiment type is available to the platform, a relatively demand-constrained market should run a $\CR$ experiment and a relatively-supply constrained market should run an $\LR$ experiment to minimize estimation error.

\noindent\textbf{Optimizing choice of treatment allocation proportion.} In other settings, however, the choice of experiment type ($\CR$ or $\LR$) may not be available to the platform.  E.g., in the preceding example where treatment increases listing description length, perhaps the difference between the two description versions is so stark that running an $\LR$ design and thus showing a given customer both the original and extended descriptions would create a disruptive customer experience, leaving $\CR$ as the only option. In another example, suppose that the platform is testing a price reduction but for legal reasons is unable to show different customers different prices for a given listing. In such a scenario, $\LR$ is the only valid design. There may be a number of such reasons why a platform is constrained to using a single design, either $\LR$ or $\CR$.

This observation leads us to also study the choice of the proportion of individuals allocated to treatment as a variable the platform might optimize to reduce estimation error.  In Section \ref{sec:optimizing_allocation}, we find that in many circumstances the treatment allocation induces a bias-variance tradeoff in both $\CR$ and $\LR$ designs. We find that the variance-optimal decision is typically to allocate an equal number of customers to treatment and control, which is what is typically done in practice. %
However, we provide theorems showing that under appropriate conditions, the bias will change monotonically in the treatment proportion, and as a result more extreme allocations reduce bias. We discuss how a platform can navigate this bias-variance tradeoff using a combination of modeling and contextual knowledge.  We also compare the relative importance of changing the experiment type compared to changing the treatment allocation.

\subsection{Related Work}

{\bf SUTVA.} The interference described in these experiments are violations of the Stable Unit Treatment Value Assumption (SUTVA) in causal inference  \cite{ImbensRubin15}, which
requires that the (potential outcome) observation on one unit should be unaffected by the particular assignment of treatments to the other units.
A large number of recent works have investigated experiment design in the presence of interference, particularly in the context of markets and social networks. 

\noindent {\bf Interference in marketplaces.} Existing work has shown that bias from interference can be large. Empirical studies \cite{Blake14, holtz2020reducing} and simulation studies \cite{Fradkin2015SearchFA} show that the size of the bias can range from one third the size to the same size as the treatment effect itself. 
Recent work has developed methods to minimize this bias using modified randomization schemes \cite{Basse16, holtz2020limiting}, experiment designs where treatment is incrementally applied to a market (e.g., small pricing changes) \cite{Wager19}, and designs that randomize on both sides of the market \cite{bajari2019double, johari2021experimental}. Specialized designs have also been designed for particular interventions, such modifications in ranking algorithms \cite{hathuc2020counterfactual}. 

In practice, platforms looking to minimize interference bias generally run clustered randomized designs \cite{chamandy16}, in which the unit of observation is changed, or switchback testing \cite{sneider19}, where the treatment is turned on and off over time. Both approaches create a large increase in variance due to the reduction in sample size, and recent work has aimed to minimize this variance in switchback designs \cite{glynn2020adaptive, bojinov2021design}. Still, due to variance concerns and ease of implementation, simpler $\CR$ and $\LR$ designs are often used in practice despite the bias that can arise. 

Our work focuses on these $\CR$ and $\LR$ designs, and how the choice of design and the proportion allocated to treatment affects bias and variance. Similar results about the choice of design and the bias were shown in a dynamic market model \cite{johari2021experimental}, although the choice of allocation and the variance is not explored in this work. %

\noindent {\bf Interference in social networks.}  A bulk of the literature in experimental design with interference considers an interference that arises through some underlying social network: e.g., \cite{Manski13, Ugander13, Athey18, Basse19, Saveski17}. In particular, \cite{pouget2019variance} and \cite{zigler2018bipartite} consider interference on a bipartite network, which is closer to a two-sided marketplace setting.

\noindent \textbf{Two-sided market model.} Our model is adopted from the work in \cite{burdett2001} which develops a clean market model that captures competition among supply and among demand, in order to study pricing. We apply the model to study experiment design. 
\section{Model}
\label{sec:model}

In this section, we describe our stylized static model for bookings in two-sided marketplaces. The supply side consists of $N$ \emph{listings} and the demand side consists of $M$ \emph{customers}. We consider a sequence of markets as we scale up both $N$ and $M$ and study the performance of experimental designs as the market grows large. 

\textbf{Listings.}
The market consists of $N$ listings that can each be matched to at most one customer.

We allow for heterogeneity of listings. Each listing $l$ has a type $\theta_l \in \Theta$ where $\Theta$ is a finite set. Let $t^{(N)}(\theta)$ denote the number of listings of type $\theta$ in the $N$'th system. For each $\theta$ assume that $\lim_{N \rightarrow \infty} t^{(N)}(\theta) / N = \tau(\theta)>0$. 
Let $\v{t}^{(N)} =\left( t^{(N)}(\theta) \right)_{\theta \in \Theta} $ and $\v{\tau} = \left( \tau(\theta) \right)_{\theta \in \Theta}$. 

For future reference, if all listings have the same type, we say that listings are {\em homogeneous}.

\textbf{Customers.}
There are $M^{(N)}$ customers in the $N$'th system. Each customer $c$ has a type $\gamma_c \in \Gamma$, where $\Gamma$ is a finite set. Let $s^{(N)}(\gamma)$ denote the number of customers of type $\gamma$ in the $N$'th system. Assume that $\lim_{N \rightarrow \infty} s^{(N)}(\gamma)/N = \sigma(\gamma)$.
Let $\v{s}^{(N)} = \left( s^{(N)}(\gamma) \right)_{\gamma \in \Gamma}$ and $\v{\sigma} = \left( \sigma(\gamma) \right)_{\gamma \in \Gamma}$. 

We scale the number of customers proportionally to the number of listings, and assume that $\lim_{N \rightarrow \infty} M^{(N)} / N = \lambda$. We refer to $\lambda$ as the ratio of \textit{relative demand} in the market. 

For future reference, if all customers have the same type, we say that customers are {\em homogeneous}.  When both customers and listings are homogeneous, we say the market is homogeneous.

\textbf{Booking procedure.}
Customers book listings through a one-shot process that captures notions of competition between listings and competition between customers. The process unfolds through a sequence of three steps. First, each customer forms a \textit{consideration set} of listings that they deem desirable. Second, customers \textit{apply} to one listing in their consideration set at random (assuming the consideration set is non-empty). Finally, each listing sees the set of customers who applied to the given listing and \textit{accepts} one customer's application, at random. This process results in a matching between customers and listings.

\textbf{Consideration sets.} 
A customer begins their experience by forming a consideration set of listings to book. For a customer of type $\gamma$ and a listing of type $\theta$, the customer has a \textit{consideration probability}  $\conprob^{(N)}(\gamma, \theta)$ of considering the listing, independent across all listings. This probability may represent factors such as whether a listing meets a customer's search criteria and the probability of a platform's recommendation system showing the listing to the customer. Let $S^{(N)}_c$ denote the consideration set for customer $c$. 

In practice, a customer will spend a limited amount of time searching through options, even as the size of the market grows; e.g., on Amazon 70\% of customers do not go past the first page \cite{clavis2015definitive}. To capture this effect, we assume that 
\begin{equation} \label{Eqn_p_ij_convergence}
N\conprob^{(N)}(\gamma, \theta) \rightarrow \conrate(\gamma, \theta)
\end{equation}
for some constant $\conrate(\gamma, \theta) \ge 0$. That is, the consideration probability a customer of type $\gamma$ has for a listing of type $\theta$ is inversely proportional to the total number of listings of type $\theta$. This ensures that the expected size of a customer's choice set $\sum_\theta t^{(N)}(\theta) \conprob^{(N)}(\gamma, \theta) $ approaches a constant as $N \rightarrow \infty$.

\textbf{Customer applications.} 
Each customer $c$ with a non-empty consideration set $S_c^{(N)}$ then chooses one listing $l \in S_c$ uniformly at random and applies to the listing.  Note that although this application is made uniformly at random, our model of heterogeneous listings can capture instances in which more attractive listings have a larger presence in the consideration set, and therefore, are more likely to be chosen by the customer.  Customers with an empty consideration set do not apply to any listings. The constraint that a customer applies to at most one listing captures {\em competition between listings} on the marketplace. A given customer $c$ becomes less likely to apply to a listing $l$ as the number of other options in their consideration set grows.

\textbf{Listing acceptances.} 
Each listing that receives a nonzero number of applications then accepts one application uniformly at random. A listing that receives no applications does not accept any customers. The resulting allocation is a matching between the set of customers and the set of listings.
Since a listing can accept at most one customer's application, the acceptance process reflects the {\em competition between customers} that arises in actual marketplaces. 

Listings do not "screen" applicants in our model; this modeling choice captures settings such as "Instant Book" on Airbnb and related features on other platforms, as well as the fact in e-commerce platforms sellers do not typically have the opportunity to screen buyers.  Of course the assumption simplifies our technical development; incorporating the opportunity for listings to screen in this model is an interesting direction for future work.

\medskip
The process outlined above models the competition between supply and competition between demand through a simple, three-step process. We now utilize this model to study experimental designs.

\section{Experimental Designs}
\label{sec:experiment_design}

Now suppose the platform considers a new feature to introduce. Before introducing this feature to the entire platform, the platform estimates the effect of this feature by running an experiment where the intervention is introduced to some fraction of the platform. Two common designs, which we focus on in this paper, are a  customer-side randomization design ($\CR$) and a listing-side randomization design ($\LR$).\footnote{We note that our model also allows for the study of more flexible experiment designs, such as the two-sided randomization design proposed in \cite{johari2021experimental} and cluster-randomized randomized designs \cite{holtz2020reducing}.} We refer to $\CR$ and $\LR$ as \textit{experiment types}.

This section discusses how we can model such experiments and the introduction of the new intervention by modifying choice probabilities, customer types, and listing types in our market model. We then use the model to study the bias and variance of these commonly used estimators.

\textbf{Interventions.}
We consider interventions that change the consideration probabilities $p(\gamma, \theta)$ for a customer including a listing in their choice set. Interventions that change these probabilities include modifications to a platform's interface, choices to show more or less information about a listing, or changes in the search and recommendation system. The changes in these consideration probabilities will propagate in the booking process to also affect the probabilities that a customer applies to a listing and the ultimate probability of a customer booking.

The intervention is binary and can either be applied or not. For a customer of type $\gamma$ and listing of type $\theta$, $p(\gamma, \theta)$ denotes the consideration probability without the intervention, and $\tilde{p}(\gamma, \theta)$ denotes the consideration probability with the intervention.

\textbf{Global treatment effect.}
We assume that the platform's primary metric of interest is the fractional number of bookings made. We focus on this metric because other metrics of interest, such as revenue, can be modeled as a function of the number of bookings. Informally, the platform then wants to measure the overall change to the fractional number of bookings made if this intervention were introduced platform-wide (global treatment) compared to a world where this intervention did not exist (global control). We call this change the \textit{global treatment effect}, or $\GTE$. 

Fix the market parameters $N, \  M^{(N)}, \ \v{s}^{(N)}$, and $\v{t}^{(N)}$. 
Formally, the global control setting is when all customers have consideration probabilities given by $\v{\conprob}=\{p(\gamma,\theta)\}_{\gamma\in\Gamma,\theta\in\Theta}$ and the global treatment setting is when all customers have consideration probabilities  $\v{\tildeconprob}=\{\tilde{p}(\gamma,\theta)\}_{\gamma\in\Gamma,\theta\in\Theta}$. Let $Q^{(N)}_\GC$ denote the (random) number of bookings made in the global control setting and $Q^{(N)}_\GT$ the (random) number of bookings made in the global treatment setting. 

We define the global treatment effect to be
\begin{equation*}
    \GTE^{(N)} = \frac{1}{N}\E \left[Q^{(N)}_\GT - Q^{(N)}_\GC \right].
\end{equation*}

\textbf{Customer-side randomized (\CR) design. }
In a $\CR$ design, the platform will decide on a treatment proportion $a_C\in(0,1)$ and perform a completely randomized design that assigns a fraction $a_C \in (0,1)$, of the $M$ customers to a treatment condition.  We let $M_1 = \lfloor a_C M \rfloor$ denote the number of customers assigned to treatment, and $M_0 = M - M_1$ denote the number of customers assigned to control.  (Our results hold for any allocations $M_0, M_1$ such that $M_0/M \to 1-a_C$ and $M_1/M \to a_C$ as $N \to \infty$.)  We denote the assignment by random variable $Z_c$ for each customer $c$, where $Z_c=1$ if they receive the intervention and $Z_c=0$ otherwise.
The customers who receive the intervention are called the \textit{treatment} group and the remaining customers are called the \textit{control} group. %

We model this assignment into treatment and control groups with an extended type space on the customers. For a customer with type $\gamma$ before the experiment is launched, we denote the type as $(\gamma, 1)$ if they are in the treatment group and $(\gamma, 0)$ if they are in the control group. A treatment customer with type $(\gamma,1)$ will have the modified consideration probability $\tilde{p}(\gamma, \theta)$ for each $\theta$, whereas the control customer with type $(\gamma,0)$ will have consideration probabilities $p(\gamma, \theta)$ for each $\theta$.

Abusing notation, we let $s^{(N)}(\gamma,1)$ and $s^{(N)}(\gamma,0)$ denote the number of customers of type $(\gamma,1)$ and $(\gamma,0)$ in the experiment, respectively, and $\sigma(\gamma,1)$ and $\sigma(\gamma,0)$ denote the limiting proportions as $N \rightarrow \infty$. 

\textbf{Listing-side randomized (\LR) design.}
Likewise, in a $\LR$ design, the platform determines a treatment proportion $a_L\in(0,1)$, and performs a completely randomized design that assigns a fraction $a_L \in (0,1)$, of the $N$ listings to a treatment condition.  We let $N_1 = \lfloor a_C N \rfloor$ denote the number of listings assigned to treatment, and $N_0 = N - N_1$ denote the number of listings assigned to control.  (Our results hold for any allocations $N_0, N_1$ such that $N_0/N \to 1-a_L$ and $N_1/N \to a_L$ as $N \to \infty$.)
For each listing $l$, let $Z_l$ denote their treatment condition where $Z_l=1$ if they receive the treatment and $Z_l=0$ otherwise.

A listing of type $\theta$ has type $(\theta,1)$ if they are assigned to treatment and type $(\theta,0)$ otherwise. For a treatment listing with type $(\theta, 1)$ any customers of type $\gamma$ will have consideration probability $\tilde{p}(\gamma, \theta)$ for that listing. For a control listing with type $(\theta,0)$, any customers of type $\gamma$ will have consideration probability $p(\gamma, \theta)$ for that listing. %

Again abusing notation, we let $t^{(N)}(\theta,1)$ and $t^{(N)}(\theta,0)$ denote the number of listings of type $(\theta,1)$ and $(\theta,0)$ in the experiment, respectively, and $\tau(\theta,1)$ and $\tau(\theta,0)$ denote the limiting proportions as $N \rightarrow \infty$.

In this modified market with the extended type space and modified consideration probabilities, bookings are made with the same three step process of consideration, application, and acceptance as described in Section \ref{sec:model}.

\textbf{Estimating the global treatment effect.}
Intuitively, the platform estimates the $\GTE$ by comparing the difference in the behavior of the treatment and control group. 

First consider a $\CR$ experiment with treatment fraction $a_C$. Let $Q_{\CR}^{(N)}(1|a_C)$ denote the number of bookings made among the treatment customers and $Q_{\CR}^{(N)}(0|a_C)$ the number of bookings made among the control customers. We present a normalized version of the commonly used difference-in-means estimator. We denote this normalized estimator $\widehat{\GTE}^\CR(a_C)$, where
\begin{equation*}
    \widehat{\GTE}_\CR^{(N)}(a_C) = \frac{M}{N} \left[\frac{1}{M_1} Q_{\CR}^{(N)}(1|a_C) - \frac{1}{M_0}Q_{\CR}^{(N)}(0|a_C) \right].
\end{equation*}
We normalize by the ratio $M/N$ to estimate the total effect on the listing side booking probability if \textit{all} customers were treated.

Now consider an $\LR$ experiment with treatment fraction $a_L$. Let $Q_{\LR}^{(N)}(1|a_L), \ Q_{\LR}^{(N)}(0|a_L)$ denote the number of bookings made among the treatment and control listings, respectively. The difference-in-means estimator for the $\LR$ design is
\begin{equation*}
    \widehat{\GTE}_\LR^{(N)}(a_L) =  \frac{1}{N_1} Q_{\LR}^{(N)}(1|a_L) - \frac{1}{N_0}Q_{\LR}^{(N)}(0|a_L) .
\end{equation*}

We will refer to difference-in-means estimators $\gtehatcr$ and $\gtehatlr$ as the (naive) $\CR$ and $\LR$ estimators in subsequent sections when context is clear.

\textbf{SUTVA and no bias.}  A key concept in causal inference is the {\em stable unit treatment value assumption} (SUTVA) \cite{ImbensRubin15}.  Informally, SUTVA requires that the outcome of a single experimental unit depends only on its own treatment assignment, and not on the treatment assignment of any other experimental units.  In experimental settings in which SUTVA holds, both $\CR$ and $\LR$ difference in means estimators are unbiased estimates of the $\GTE$, that is:
\[
\E\left[\widehat{\GTE}^{(N)}_\CR\right] = \E\left[\widehat{\GTE}^{(N)}_\LR\right] = \GTE^{(N)}.
\]
As we will discuss, however, in the market model with customer and listing competition, SUTVA does not hold and the estimators will be biased in general. %

\section{Large Market Setting}
\label{sec:limiting_behavior}

In the previous section, we described a model where experiment interventions exogenously change consideration probabilities. This change in turn creates endogenous changes in the application behavior and the number of bookings made on the platform, the latter of which is the metric of interest and serves as the basis for the $\GTE$ and the $\CR$ and $\LR$ estimators. In this section, we characterize the application behavior and booking behavior on the platform and their dependence on consideration probabilities and other model primitives. This characterization will allow us to study the bias and variance of $\CR$ and $\LR$ estimators in the next section.

The competition that arises between customers and between listings creates an interdependence in the market that complicates the analysis of applications and bookings. In order to analyze the system, we study the behavior as the market grows large, that is, as the number of listings $N \rightarrow \infty$. Recall that we scale the number of customers such that $M^{(N)}/N \rightarrow \lambda$, where $\lambda$ is the relative demand in the market. When we take $N \rightarrow \infty$, we fix this ratio of relative demand and scale up both the supply side and demand side.

Consider a heterogeneous market where listing and customer type distributions approach the vectors $\bm{\tau}=(\tau(\theta))_{\tinT}$ and $\bm{\sigma}=(\sigma(\gamma))_{\ginG}$, respectively. We analyze the probability of booking by following the three steps in the booking procedure in Section~\ref{sec:model}: consideration, application, and acceptance. 

\textbf{Consideration sets.}
We first examine how the consideration sets are formed. As we have defined in the model, let $\conprob^{(N)}(\gamma,\theta)$ denote the probability that a (fixed) listing of type $\theta$ is included in a (fixed) type-$\gamma$ customer's consideration set, and we have
$   N \conprob^{(N)}(\gamma,\theta)\to \conrate(\gamma,\theta) $ , %
for all pairs of customer and listing. We refer to $\conrate(\gamma,\theta)$ as the limit \emph{rate} of consideration.
The inclusion of listings into customers' consideration sets is mutually independent.
For a customer $c$ of type $\gamma_c$, the number of type-$\theta$ listings in their consideration set $S_c$ follows a binomial distribution $\Binom(t(\theta), \conprob(\gamma_c,\theta))$, which, as is well known, converges to a Poisson distribution with rate parameter $\tau(\theta) \conrate(\gamma_c,\theta)$ as the market grows large. 

\textbf{Customer applications.}
Next, we describe the formation of application from consideration sets.
Let $\appprob(\gamma,\theta;\v{t})$
denote the probability that a customer of type $\gamma$ applies to a certain listing of type $\theta$. Again, this value does not depend on our choice of type $\gamma$ customer and type $\theta$ listing, as all heterogeneity in our model dependent only on types. Clearly, the applications are not mutually independent, because of the constraint that each customer can apply to at most one listing. However, from each listing's point of view, the applications they receive from all the customers are mutually independent.
For a listing $l$ of type $\theta_l$, the number of applications they receive from type $\gamma$ customers follows  $\Binom(s(\gamma), \appprob(\gamma,\theta_l;\v{t}))$. This approaches a Poisson distribution as $N\to\infty$ as long as $s(\gamma) \appprob(\gamma,\theta_l;\v{t})$ converges to a constant limit, which we will show below. %

\textbf{Bookings.}
In a similar way, we formulate the emergence of the final matching out of the applications. Let $\bookprob(\gamma,\theta;\v{t},\v{s})$ denote the probability that customer $c$ of type $\gamma$ applies to a certain listing $l$ of type $\theta$ and is accepted.

\medskip
Utilizing the property that the binomial distribution converges to the Poisson distribution, we then establish the following lemma on the behavior of applications and bookings in a large market. 
The proof is given in Appendix~\ref{Appendix:proofs}.

\begin{restatable}{lemma}{lemmaLimAppBookProb}
\label{Prop_limit_app_and_book_prob}
For any customer type $\gamma$ and any listing type $\theta$, as $N \to \infty$,
\begin{equation}\label{Eqn_limit_app_prob:Prop_limit_app_and_book_prob}
    N \appprob^{(N)}(\gamma,\theta;\v{t}^{(N)}) \to \apprate(\gamma,\theta),
\end{equation}
\begin{equation}\label{Eqn_limit_book_prob:Prop_limit_app_and_book_prob}
    N \bookprob^{(N)}(\gamma,\theta;\v{t}^{(N)},\v{s}^{(N)}) \to \bookrate(\gamma,\theta),
\end{equation}
with the application rate matrix $\appratematrix = \{\apprate(\gamma,\theta)\}_{\ginG,\tinT}$ and booking rate matrix $\bookratematrix = \{\bookrate(\gamma,\theta)\}_{\ginG,\tinT}$ given by
\begin{equation}\label{Eqn_def_phi_app:Prop_limit_app_and_book_prob}
    \apprate(\gamma,\theta) = \conrate(\gamma,\theta) \frac{1-\exp(-\v{\tau}\cdot\v{\conratematrix}(\gamma,\cdot))}{\v{\tau}\cdot\v{\conratematrix}(\gamma,\cdot)},
\end{equation}
\begin{equation}\label{Eqn_def_phi_book:Prop_limit_app_and_book_prob}
    \bookrate(\gamma,\theta) = \apprate(\gamma,\theta) \frac{1-\exp(-\lambda\bsigma\cdot\v{\appratematrix}(\cdot,\theta))}{\lambda\bsigma\cdot\v{\appratematrix}(\cdot,\theta)},
\end{equation}
where $\bm{A}(i,\cdot)$ and $\bm{A}(\cdot, j)$ denote the $i$'th row and $j$'th column (as vectors) of a matrix $A$, respectively.\footnote{In the case when $\v{\tau}\cdot\v{\conratematrix}(\gamma,\cdot)=0$ for some $\ginG$ or $\v{\sigma}\cdot\v{\appratematrix}(\cdot,\theta)=0$ for some $\tinT$, \eqref{Eqn_def_phi_app:Prop_limit_app_and_book_prob} and \eqref{Eqn_def_phi_book:Prop_limit_app_and_book_prob} will take value zero by continuity.}
\end{restatable}

\begin{remark}
So far, we have been explicit about showing the dependencies between variables, namely on $N$ and how the application and booking probabilities $q$ and $r$ depend on the distributions of listing and customer types $\v{t}$ and $\v{s}$. In subsequent discussions, we may omit the $(N)$ superscripts and the dependency of $q$ and $r$ on $\v{t}$ and $\v{s}$; these dependencies should always be implicitly assumed.
\end{remark}

To compactify the expressions, we introduce the following function $F:[0,\infty) \to \R^+$ defined as
\[
    F(x) = \frac{1-\exp(-x)}{x} \;\text{ for }\; x>0,
\]
with $F(0)=1$ by continuity.
It is straightforward to verify that $F(x)\in(0,1]$ for any $x\ge 0$ and is monotonically decreasing on $[0,\infty)$. To interpret this, imagine a $\Poisson(x)$ sequence of jobs arriving at a server that serves exactly one job (if any) during each unit time interval (other arrivals are dropped). Then, $F(x)$ is the probability that an incoming job will be served.

Using this notation, we can rewrite equations \eqref{Eqn_def_phi_app:Prop_limit_app_and_book_prob} and \eqref{Eqn_def_phi_book:Prop_limit_app_and_book_prob} as
\begin{equation}\label{Eqn_redef_app_rate_using_F}
    \apprate(\gamma,\theta) = \conrate(\gamma,\theta) F(\v{\tau}\cdot\v{\conratematrix}(\gamma,\cdot)),
\end{equation}
\begin{equation}\label{Eqn_redef_book_rate_using_F}
    \bookrate(\gamma,\theta) = \apprate(\gamma,\theta) F(\lambda\bsigma\cdot\v{\appratematrix}(\cdot,\theta)).
\end{equation}
Equation \eqref{Eqn_redef_app_rate_using_F} gives an interpretation for the term $F(\v{\tau}\cdot\v{\conratematrix}(\gamma,\cdot))$ as the average conversion probability of consideration to applications for a type-$\gamma$ customer. Similarly, from equation \eqref{Eqn_redef_book_rate_using_F}, $F(\lambda\bsigma\cdot\v{\appratematrix}(\cdot,\theta))$ can be interpreted as the application-to-booking conversion probability for listings of type $\theta$.

An immediate corollary of the previous lemma is the convergence of the global booking rate (of listings) to the following limit. 

\begin{restatable}[Limit of booking rate]{corollary}{corLimBookingRate}
\label{Corollary_limit_booking_rate}
Recall that $Q=Q^{(N)}$ denotes the total number of bookings. As $N \to \infty$,
\[
   N^{-1} \E \left[Q\right] \to \lambda \bm{\sigma}^T \bookratematrix \bm{\tau}.
\]
\end{restatable}

\section{Bias and Variance of LR and CR Estimators}
\label{sec:bias_and_var_of_lr_cr}

As discussed at the end of Section \ref{sec:experiment_design}, in general the difference-in-means estimators used with CR and LR designs are biased.  This is a well-known observation in the literature that is traceable to the fact that each experiment design creates interference through common interactions with the opposite side of the market.  In a CR design, {\em both} treatment and control customers interact with the same supply, and thus, if a treatment customer is able to book a listing, that mechanically implies a reduction in effective supply for control customers.  Similarly, in a LR design, {\em both} treatment and control listings interact with the same customers, and thus, if a treatment listing receives applications from customers, it mechanically means lower effective demand for control listings.  These interactions lead to interference in a statistical sense, and bias the resulting estimators\cite{Blake14, Fradkin2015SearchFA, holtz2020reducing, Wager19, bajari2019double, johari2021experimental}.  In this section, we quantify this bias and variance of estimators in such conditions.

\subsection{$\CR$ and $\LR$ bias}
\label{ssec:limiting_lr_cr_bias}
We apply the results on the limiting behavior of the booking rates to characterize the behavior of the $\CR$ and $\LR$ designs and estimators. Before the intervention is introduced, as stated in Section \ref{sec:model} the customers have type space $\Gamma$ and the listings have type space $\Theta$, with the number of customers of each type given by $\v{t}^{(N)}$ and the number of listings of each type given by $\v{s}^{(N)}$. Recall that the intervention changes the consideration probability matrices from $\v{\conprob}$ to $\v{\tildeconprob}$. 

First we calculate the $\GTE$. This result directly follows by applying Corollary \ref{Corollary_limit_booking_rate} to the global treatment setting and global control setting. With treatment applied to the entire market, i.e., under global treatment, we define application probability matrix $\tildeappratematrix = \{\tildeapprate(\gamma,\theta)\}_{\gamma\in\Gamma,\theta\in\Theta}$ and booking probability matrix $\tildebookratematrix = \{\tildebookrate(\gamma,\theta)\}_{\gamma\in\Gamma,\theta\in\Theta}$ analogous to \eqref{Eqn_def_phi_app:Prop_limit_app_and_book_prob} and \eqref{Eqn_def_phi_book:Prop_limit_app_and_book_prob}.

\begin{corollary}[Limit of GTE]\label{Corollary_limit_gte}
As $N \rightarrow \infty$,
\[
    \GTE^{(N)} = \frac{1}{N}\E\left[Q_{\GT}^{(N)} - Q_{\GC}^{(N)}\right] \to \lambda \bsigma^T \left(\tildebookratematrix - \bookratematrix\right) \v{\tau}.
\]
\end{corollary}

Now consider an experimental setting where the platform allocates either listings or customers to treatment and control groups. 
Recall that the allocation of a fraction of customers or listings can be considered as a modification of customer or listing types, respectively. Hence, we can similarly establish limits for the expectations of our $\CR$ and $\LR$ estimators through an application of Lemma \ref{Prop_limit_app_and_book_prob}. We provide the full characterization for general, heterogeneous markets in Proposition \ref{Prop_limit_estimator_expectation} (Appendix \ref{Appendix:proofs}). For brevity, we discuss the result in a homogeneous market, though all intuition extends to a heterogeneous market. When the market is homogeneous, Proposition~\ref{Prop_limit_estimator_expectation} reduces to the following.

\begin{corollary}\label{Corollary_homogen_limit_estimator_expectation}
Assume that the listings and customers are homogeneous, with limit consideration rate $N^{-1}\conprob \to \conrate$ under control and $N^{-1}\tildeconprob \to \tildeconrate$ under treatment. Then
\begin{equation}\label{Eqn_homogen_limit_cr_est_expectation:Corollary_homogen_limit_estimator_expectation}
    \lim_{N\to\infty} \gtehatcr(a_C) = \lambda (\tildeconrate F(\tildeconrate) - \conrate F(\conrate)) F(\lambda (a_C\tildeconrate F(\tildeconrate) + (1-a_C)\conrate F(\conrate))),
\end{equation}
\begin{equation}\label{Eqn_homogen_limit_lr_est_expectation:Corollary_homogen_limit_estimator_expectation}
    \lim_{N\to\infty} \gtehatlr(a_L) = \exp(- \lambda \tildeconrate F(a_L\tildeconrate + (1-a_L)\conrate)) - \exp(- \lambda \conrate F(a_L\tildeconrate + (1-a_L)\conrate)).
\end{equation}
\end{corollary}

To understand how these expressions capture competition and interference in the experiments, compare the expressions for the limiting $\GTE$ in the homogeneous case, which by Corollary~\ref{Corollary_limit_gte} is
\[
    \lim_{N\to\infty} GTE = \lambda (\tildeconrate F(\tildeconrate) F(\lambda \tildeconrate F(\tildeconrate)) - \conrate F(\conrate) F(\lambda \conrate F(\conrate))) = \exp(-\lambda \conrate F(\conrate)) - \exp(-\lambda \tildeconrate F(\tildeconrate)).
\]

Recall that in our model, the conversion probabilities from consideration to applications and from applications to bookings can both be expressed with the function $F$. Consider the homogeneous example, where $F(\conrate)$ is the consideration-to-application conversion probability under global control, and $F(\lambda \apprate)$ with $\apprate=\phi F(\conrate)$ is the application-to-booking conversion probability.

In an $\LR$ experiment when $a_L$ fraction of the listings are treated and treated listings have a consideration rate $\tildeconrate > \conrate$, a customer's consideration set will include a mix of control and treatment listings: that is, the intensity of consideration for each customer increases to $(1-a_L)\conrate + a_L\tildeconrate$. The rate at which consideration of listings converts to applications becomes $F((1-a_L)\conrate + a_L\tildeconrate) < F(\conrate)$. In other words, a control listing now must compete with treatments listings for customer applications relative to the global control condition. %
A similar argument applies to the treated listings.
This ``mixing''  in the consideration set exactly reflects the competition or interference between treatment and control listings in $\LR$ experiments, and is the source of the resulting estimation bias.
Assuming $\tildeconrate > \conrate$, the overall consideration rate is lower for the customers in $\LR$ experiments than in global treatment, and thus the conversion probability in $\LR$ experiments must be higher than that in global treatment, $F(\tildeconrate)$;
on the other hand, it must be lower than that in global control, $F(\conrate)$. 
A treatment listing is more likely to receive an application from a customer in $\LR$ experiments than in global treatment, conditioning on being considered, and similarly control listings are less likely to receive applications in the experiment than in global control. 
In other words, the mixed environment in $\LR$ experiments causes an undue advantage to the treatment listings, making them better off than they would be in global treatment, while making the control listings worse off than they would be in global control.

Similarly, in global control, the asymptotic application-to-bookings conversion probability is given by $F(\lambda\apprate)$, where $\apprate$ is the application intensity as discussed above. In a $\CR$ experiment, however, each listing will receive applications from a mix of treatment and control customers, now with the blended intensity of $(1-a_C)\apprate + a_C \tildeapprate > \apprate$. This means that each of the applications from control customers now have to compete for acceptance with some additional number of applications from treated customers, intensifying competition relative to global control. Meanwhile, each treatment customer will experience less intense competition in the $\CR$ experiment than they would under global treatment. 
The intensity of competition is captured by the rate at which an application is accepted $F(\lambda ((1-a_C) \apprate + a_C \tildeapprate))$. This acceptance rate is lower than that in global control, due to more applications and hence more intense competition among the customers, and higher than that in global treatment for the analogous reason. Such mixed competition between the control and treatment customers creates an advantage for treatment customers and meanwhile harms the control customers.

The discussion of how bias arises motivates the following lemma, which asserts that under certain condition, the competition effect in $\CR$ and $\LR$ leads to a positive bias in the naive estimators. In subsequent discussion, we say that an intervention is \emph{multiplicative} in consideration probabilities $\conprob(\gamma,\theta)$ if there exists $\alpha > 0$ such that $\tildeconprob(\gamma,\theta) = \alpha \conprob(\gamma,\theta)$ for all customer types $\ginG$ and listing types $\tinT$. We call a multiplicative intervention \emph{positive} or \emph{upward} if $\alpha > 1$, and \emph{negative} or \emph{downward} if $\alpha < 1$.

\begin{restatable}{lemma}{lemmaUpwardBias}
\label{Lemma_upward_bias}
When the intervention is positive and multiplicative in consideration probabilities, then the bias is asymptotically positive for both $\CR$ and $\LR$ estimators, i.e.
\[
    \lim_{N\to\infty} \left(\E\left[\gtehatcr^{(N)}(a_C)\right] - GTE^{(N)}\right) > 0 \;\text{ and }\; \lim_{N\to\infty} \left(\E\left[\gtehatlr^{(N)}(a_L)\right] - GTE^{(N)}\right) > 0
\]
for any fixed $a_C,a_L\in(0,1)$. (Analogously, when the intervention is negative and multiplicative, the bias is asymptotically negative for both $\CR$ and $\LR$ estimators.)
\end{restatable}

The previous theorem relies on the assumption that the intervention has an equal and multiplicative effect on $\conrate(\gamma,\theta)$ for all $\ginG$ and $\tinT$. In general, of course, if treatment has different effects on different pairs of listing and customer types, the resulting sign of the treatment effect may depend on a complex way on the model primitives.  In particular, the bias of both $\CR$ and $\LR$ experiments will be affected by market conditions, such as the type distribution of listings and customers and the relative demand $\lambda$. The bias is also affected by the  magnitude and sign of $\tildeconprob - \conprob$, i.e., the lift that the intervention has on the consideration probability at the customer-listing pair level.  Nevertheless, in our subsequent development the case of multiplicative effects will be a valuable benchmark within which to develop intuition.

\subsection{Characterization of $\LR$ and $\CR$ variance}
\label{ssec:char_lr_cr_variance}

The simplicity of our model allows us to analyze the asymptotic behavior of variance in $\LR$ and $\CR$ experiments. If the bookings are made independently, i.e. the Bernoulli random variables $Y_l$ indicating whether listing $l$ is booked are mutually independent, then the standard error can be fully characterized by a classic binomial model. However, in the presence of competition in the market, we also need to account for the negative correlation between $Y_l$ and $Y_{l'}$ for different listings $l$ and $l'$. 

In Appendix \ref{sec_variance_appendix}, we study the variance of the sampling distribution of $\LR$ and $\CR$ estimators, in a homogeneous market.  The tractability of our model allows us to derive an explicit expression for the variance of the estimators in a large market, incorporating the correlation terms. We show that the expressions indeed correspond closely with the variance obtained through simulations. We then leverage this analysis and the resulting expressions for numerics to qualitatively study variance-optimal designs in Section~\ref{subsec:effect_of_allocation_on_var}.

\section{Optimizing experiments: Experiment type}
\label{sec:optimizing_exp_type}

We now turn our attention to two levers that a platform has when designing experiments: the choice of experiment type ($\CR$ or $\LR$) and the treatment allocation for a given experiment type ($a_C$ and $a_L$). In this section, we focus on the choice of experiment type for the platform.  We find that the bias-optimal type depends on market balance, with $\CR$ bias diminishing as relative demand $\lambda \to 0$ and $\LR$ bias diminishing as relative demand $\lambda \to \infty$.  (A similar result for the behavior of bias in market extremes was found in a dynamic market model in \cite{johari2021experimental}.) 
Moreover, we show through simulations that the bias-optimal type often coincides with the variance-optimal type, or that the two types have similar variance. In short, there is no pronounced bias-variance tradeoff in the choice of experiment type. 

\subsection{Bias of $\CR$ versus $\LR$}
Using Proposition    \ref{Prop_limit_estimator_expectation}, we can explicitly characterize the $\CR$ and $\LR$ bias as a function of the relative demand, customer and listing type distributions, and pre-treatment and post-treatment consideration probabilities.  In particular, we can show that the $\CR$ estimator becomes unbiased as the relative demand diminishes and the $\LR$ estimator becomes unbiased as the relative demand increases.

\begin{restatable}[Unbiasedness in market extremes]{theorem}{thrmUnbiasedExtremes}\label{ThrmUnbiasedExtremes}

\begin{enumerate}
    \item Consider a sequence of markets where $\lambda \to 0$.
    Along this sequence, for any $a_C \in (0,1)$ we have $\lim_{N\to\infty}\lambda^{-1}\big(\gtehatcr(a_C)-\GTE\Big) \to 0$ as $\lambda\to 0$. That is, the asymptotic bias of the $\CR$ estimator approaches 0 as the relative demand decreases. The asymptotic bias of the $\LR$ estimator $\lim_{N\to\infty}\lambda^{-1}\big(\gtehatlr(a_L)-\GTE\big)$, however, is bounded away from zero for any $a_L\in(0,1)$.
    \item Consider a sequence of markets where $\lambda \to \infty$. Along this sequence, for any $a_L \in (0,1)$ we have $\lim_{N\to\infty}\big(\gtehatlr(a_L) - \GTE\big) \to 0$. That is, the asymptotic bias of the $\LR$ estimator approaches 0 as the relative demand increases. The asymptotic bias of the $\CR$ estimator $\lim_{N\to\infty}\big(\gtehatcr(a_C)-\GTE\big)$, however, is bounded away from zero for any $a_C\in(0,1)$.
\end{enumerate}
\end{restatable}
\begin{remark}
In the first part of the proposition we normalize by $\lambda$, because if $\lambda \to 0$ then both the $\GTE$ and both estimators go to zero mechanically.  
In the second part of the proposition, however, the $\GTE$ does not mechanically go to zero, so the result does not include the normalization by $\lambda$.
\end{remark}

Intuitively, bias in the $\CR$ estimate arises when treatment and control customers apply to the same listing and thus compete with each other. Bias in the $\LR$ estimate arises when customers consider both control and treatment listings in their consideration set, creating competition between listings. In a demand constrained market as $\lambda \to 0$, there are few enough customers that customers are unlikely to apply to the same listings and so the $\CR$ estimator is unbiased. However, the existing customers will still have multiple listings in their consideration sets, so in an $\LR$ experiment, treatment and control listings will still compete and the $\LR$ estimate will be biased.

In over-demanded market as $\lambda \to \infty$, many customers will apply to a given listing, creating competition between customers, and thus biasing the $\CR$ estimator. The competition between listings, created by customers comparing multiple listings in their consideration set, persists in this extreme as well. However, as the number of customers grows, all listings will receive an application (and thus be booked), regardless of the competition created by multiple listings appearing in a customer's consideration set and of the treatment condition. Hence, both the $\GTE$ and the $\LR$ estimate approach zero, and thus the $\LR$ estimate becomes unbiased.

\subsection{Optimizing bias and variance}
Using simulations, we analyze the variance of the $\CR$ and $\LR$ estimators, and we find that the choice of experiment type does not induce a significant bias-variance tradeoff. 

\begin{figure}[h]
    \centering
    \includegraphics[scale=.55]{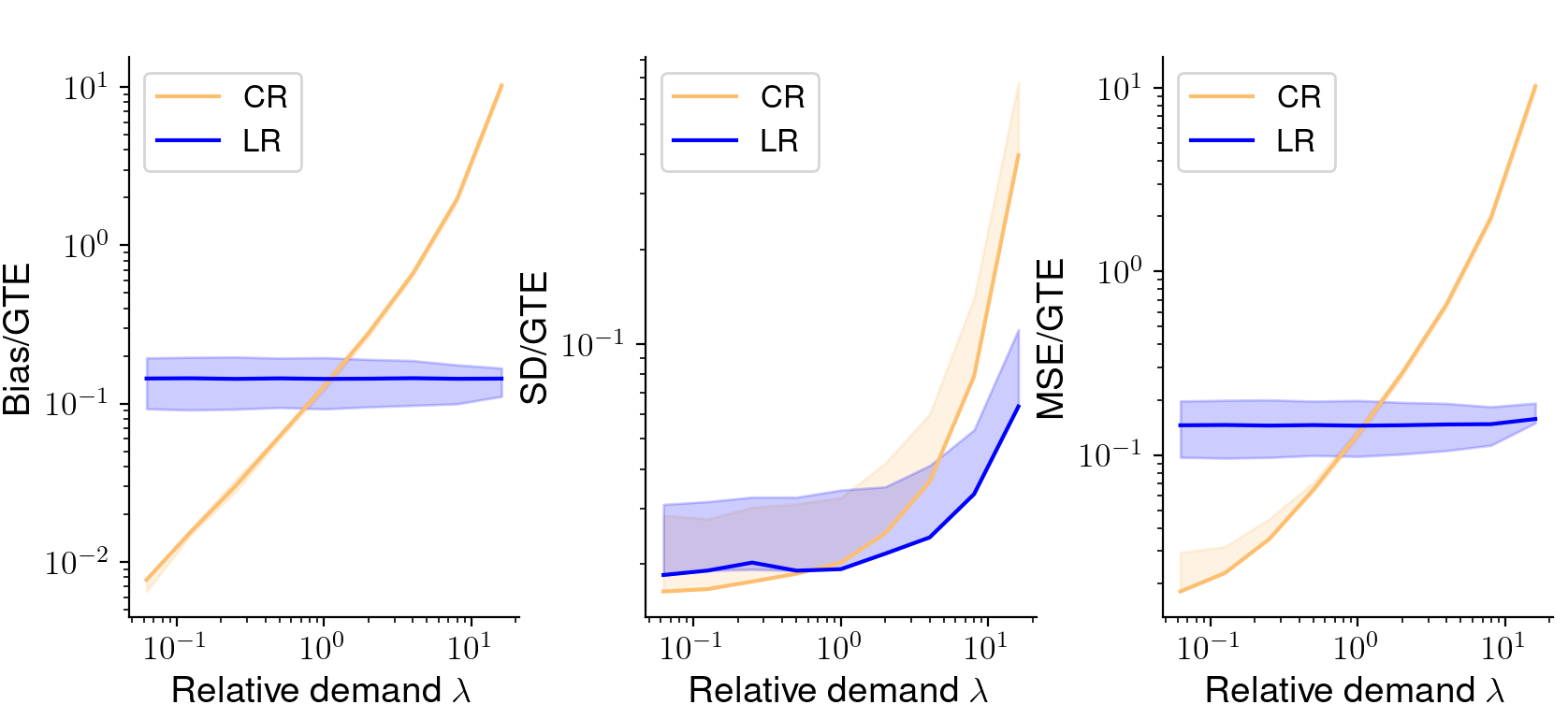}
    \caption{Behavior of bias, SD, and MSE as relative demand $\lambda$ changes. Simulations in homogeneous markets with $M=2^{22} \approx 4.1$ million customers, and number of listings $M$ varying between $2^{16} \approx 66$ thousand to $2^{28} \approx 268$ million at each integer power of 2. Solid curves represent a 50\% allocation to treatment. Shaded bands represent all possible values achieved with a treatment allocation in the range $[0.1, 0.9]$. Consideration probabilities are fixed across the market and are chosen such that the booking probability in a balanced market (i.e., $M = N$) is $20\%$ under global control and $22\%$ under global treatment, with $\conrate = N\conprob \approx 0.253$ and $\tildeconrate = N\tildeconprob\approx 0.286$.}
    \label{fig:vary_lambda_bias_sd_mse}
\end{figure}

In Figure \ref{fig:vary_lambda_bias_sd_mse}, we consider a homogeneous market and compare the performance of $\LR$ and $\CR$ estimators as $\lambda$ varies. The parameters are calibrated to reflect reasonable booking probabilities and treatment effects: in a balanced market, 20 percent of listings are booked in global control and 22 percent in global treatment. 

We see that the $\LR$ and $\CR$ estimators have similar variance for $\lambda<1$. In this range, the $\CR$ design has lower bias, and so a platform aiming to minimize $\MSE$ should run a $\CR$ experiment. For $\lambda>1$, an $\LR$ experiment leads to lower variance than a $\CR$ experiment. Thus for a market with higher demand, an $\LR$ experiment minimizes both bias and variance, and so minimizes the $\MSE$. In Appendix \ref{app:figures}, we find that these observations hold true in scenarios with varying $\GTE$ and heterogeneity. 

It is interesting to note that for the $\LR$ estimator, although the bias goes to zero in the supply-constrained limit in an {\em absolute} sense (from Theorem \ref{ThrmUnbiasedExtremes}), it does not go to zero in a {\em relative} sense (normalized by $\GTE$) cf.~Figure \ref{fig:vary_lambda_bias_sd_mse}.  In fact, the relative bias remains fairly flat for the $\LR$ estimator.  By contrast, both the absolute and relative bias of the $\CR$ estimator approach zero in the demand-constrained limit.

In short, if the choice of experiment type is available to the platform, then choosing the bias-minimizing design does not increase the variance of the design, and in some cases may even decrease the variance. A relatively demand constrained market should run a $\CR$ experiment and a relatively supply constrained market should run an $\LR$ experiment.

\section{Optimizing experiments: Treatment allocation }
\label{sec:optimizing_allocation}

Once the experiment type is fixed, the platform must choose what proportion of individuals (either listings or customers) to randomize to treatment. Typically, platforms will randomize individuals to receive treatment and control with equal probability; in settings without interference and with independent observations, this 50-50 split (i.e., treatment allocation of $0.5$) decreases the variance of the estimator and increases the statistical power of the experiment.  In our setting, however, there are two complicating factors: First,  $\CR$ and $\LR$ estimators are typically biased, and this bias can vary with the treatment allocation.  Second, in two-sided markets interference also creates correlation between observations, so the behavior of variance with treatment allocation is not immediately obvious.  In this section we investigate these issues.

We show that the choice of treatment allocation induces a bias-variance tradeoff. 
First, we show that the bias-minimizing treatment allocation probability often lies at an endpoint of the range $[0,1]$, for both $\CR$ and $\LR$. The variance-minimizing allocation, however, is roughly $0.5$ (for reasonably small treatment effect sizes). 
We discuss the factors contributing to whether a platform should optimize for bias or optimize for variance. Overall, we find through simulations that even though it is not always optimal, typically a 50-50 split is a relatively robust choice of allocation for minimizing the $\MSE$ in many practical scenarios. Finally, we compare the relative importance of optimizing the experiment type and optimizing the treatment allocation. While the treatment allocation can offer improvements, the effect of choosing the correct experiment type greatly outweighs the smaller gains from optimizing the allocation.

We note that in some practical settings, if an intervention is deemed too risky or expensive, or if statistical power is not a concern, then a platform may start with a relatively small initial treatment allocation. The platform then employs a ``ramp up'' process \cite{xu2018sqr}, where it waits to see initial effects on the metric of interest before incrementally increasing the treatment proportion. This process of waiting and increasing may happen several times until the intervention is eventually introduced to the entire population. Our work in this section is also valuable to platforms implementing such strategies, as the proportion allocated to treatment changes the bias in the resulting estimator. (For further discussion see Section \ref{sec:discussion}.)

\subsection{Effect of treatment allocation on $\CR$ bias}

We show for a subclass of intervention types that bias is monotonic in the treatment proportion, though we conjecture that the result holds more generally. Recall that we call an intervention \textit{multiplicative} if there exists an $\alpha$ such that $\v{\tildeconprob}=\alpha \v{\conprob}$. That is, the intervention has the same multiplicative lift on the consideration probability across all customer and listing types. The following result shows that for $\CR$ and a multiplicative intervention, the bias is decreasing in $a_C$ if $\alpha>1$ and increasing in $a_C$ if $\alpha<1$. For a $\CR$ design with treatment probability $a_C$, define the %
asymptotic bias of the experiment to be $B_\CR(a_C)$ where 
\[
    B_\CR(a_C) = \lim_{N \rightarrow \infty} \left( \E\left[\widehat{\GTE}_\CR^{(N)}(a_C)\right] - \GTE^{(N)} \right).
\]

\begin{restatable}{theorem}{thrmCROptBiasAC}
\label{Thm_opt_bias_ac}
Suppose that the intervention is multiplicative on $\v{\conprob}$ with parameter $\alpha$. If $\alpha>1$, then the asymptotic bias of the $\CR$ estimator i.e., $|B_\CR(a_C)|$, decreases as $a_C$ increases. If $\alpha<1$, then the asymptotic bias of the $\CR$ estimator increases as $a_C$ increases.
\end{restatable}

This result may appear surprising, and so we provide some intuition to reveal why the bias of the $\CR$ estimator is decreasing in $a_C$.  Consider an example where the market is homogeneous with $M_1$ treatment customers, $M_0$ control customers, and $N$ listings.  Treatment (resp., control) customers choose a listing to consider with probability $\tilde{\conprob}$ (resp., $\conprob$), with $\tilde{\conprob} > \conprob$.  Now consider adding a new customer to this market.  Because we know that (a) the $\GTE$ is positive and (b) the $\CR$ estimator overestimates the $\GTE$ for any treatment allocation, the bias will go down if we reduce the value of the estimator.  Thus we want to add the customer to the group that leads the greatest reduction in the value of the estimator.

Let $B_1$ be the expected number of bookings in the treatment group, and let $B_0$ be the expected number of bookings in the control group.  Let $Y_0$ be the probability the new customer books as a control customer, and let $Y_1$ be the probability the new customer books as a treatment customer.  The key observation is that: $Y_1 < B_1/m_1$ and $Y_0 < B_0/m_0$; in other words, the new customer is less likely to book than the average booking rate of existing customers in either group.  This is because in order to book, the new customer has to apply to an entirely new listing that previously had no applications.  If the new customer is in treatment (resp., control), this is strictly less likely than any of the existing treatment (resp., control) customers.  The $\CR$ estimator is $B_1/M_1 - B_0/M_0$.  If we add the customer to the treatment group, the new estimator is $(B_1 + Y_1)/(M_1 + 1) - B_0/M_0$, while if we add the customer to the control group, the new estimator is $B_1/M_1 - (B_0 + Y_0)/(M_0 + 1)$.  It is straightforward to verify that the estimator, and thus the bias, is smaller if we add the customer to the treatment group.

We conjecture that a more general result than Theorem \ref{Thm_opt_bias_ac} holds for a class of interventions beyond multiplicative interventions, as long as the amount of heterogeneity across the differences $\tilde{\conprob}(\gamma, \theta) - \conprob(\gamma, \theta)$ is sufficiently small across pairs of $\gamma$ and $\theta$.  If $\tilde{\conprob}(\gamma, \theta) - \conprob(\gamma, \theta)$ is very heterogeneous across customer and listing types, the conjecture may fail; we have observed this in some examples where the change in consideration probability is positive for some pairs of customer and listing types, and negative for others. For interventions that platforms expect to have very heterogeneous effects, the direction of the bias is not obvious upfront.

Though the treatment allocation $a_C$ affects the magnitude of the bias in a $\CR$ experiment, we can show that the choice of $a_C$ does not affect the bias \textit{too} much. More specifically, we can bound the maximum difference in bias with respect to the choice of $a_C$. 
For a $\CR$ experiment, we call the difference between the highest possible and lowest possible bias attained by varying $a_C$
$$\sup_{a_C\in(0,1)} B_\CR(a_C) - \inf_{a_C\in(0,1)} B_\CR(a_C)$$ the \textit{$\CR$ bias differential} in $a_C$. 

From \eqref{Eqn_partial_derivative_ac:Thm_opt_bias_ac}, we immediately observe a bound on the $\CR$ bias differential given by the Lipshitz coefficient of $B_{\CR}$, as stated in the following corollary.

\begin{corollary}
In a $\CR$ experiment, the bias differential for the naive estimator is bounded by
\[
    \sup_{a_C\in(0,1)} B_\CR(a_C) - \inf_{a_C\in(0,1)} B_\CR(a_C) \leq \lambda^2 \left( \sup_{\ginG,\tinT}\left[ \tildeconrate(\gamma,\theta)-\conrate(\gamma,\theta) \right]\right)^2.
\]
\end{corollary}

Note that this result holds for any intervention, even when the intervention is not multiplicative on $\Phi$. This bound depends on the relative demand $\lambda$ as well as the size of the lift that the intervention has on the consideration probabilities. Figures \ref{fig:vary_lambda_bias_sd_mse} and \ref{fig:lrcr_bias_vary_treatment} show how the bias differential behaves when varying demand $\lambda$ and the size of the treatment effect.

\subsection{Effect of treatment allocation on $\LR$ bias}

For an $\LR$ design with treatment probability $a_L$, define the asymptotic bias to be 
$B_\LR(a_L)$ where 
$$B_\LR(a_L) = \lim_{N \rightarrow \infty} \left( \E\left[\widehat{\GTE}_\LR^{(N)}(a_L)\right] - \GTE^{(N)} \right).$$

We similarly show that for the $\LR$ design, the bias is monotonically changing in $a_L$. However, the result here is more subtle: for a fixed intervention, the bias may be monotonically increasing or monotonically decreasing in $a_L$, depending on both the sign of the change in the consideration probability and the relative demand $\lambda$.

We prove this result for a homogeneous market (where trivially the intervention is multiplicative).  
\begin{restatable}{theorem}{thrmLROptBiasAL}
\label{thm:opt_bias_aL}
Consider a homogeneous market (with only one type of listings and customers), where the consideration probability changes from $\conprob$ to $\tilde{\conprob}$. Then the asymptotic bias of the LR estimator, as a function of $a_L$, has no local minimum inside $(0,1)$. In other words, it achieves infimum either as $a_L\to 0$ or as $a_L\to 1$. Furthermore, there exists a $\lambda^*$ such that 
\begin{enumerate}[label=\alph*)]
    \item If $\tilde{\conprob}>\conprob$, then  $B_\LR(a_L)$ is decreasing in $a_L$ for $\lambda < \lambda^*$ and  $B_\LR(a_L)$ is increasing in $a_L$ for $\lambda>\lambda^*$ . 
    \item If $\tilde{\conprob}<\conprob$, then $B_\LR(a_L)$ is increasing in $a_L$ for $\lambda < \lambda^*$  and $B_\LR(a_L)$ is decreasing in $a_L$ for $\lambda>\lambda^*$.
\end{enumerate}
\end{restatable}

We conjecture that a similar, more general statement holds for multiplicative interventions in a heterogeneous market as well, and also for interventions where the lifts on the consideration probabilities $\conprob(\gamma, \theta)$ are not ``too'' heterogeneous across $\gamma$ and $\theta$.  More specifically, we conjecture that on a broader class of interventions, there exists a cutoff $\lambda^*$ such that when the $\GTE>0$, the $\LR$ asymptotic bias is decreasing in $a_L$ for $\lambda<\lambda^*$ and increasing in $a_L$ for $\lambda>\lambda^*$, and vice versa for $\GTE<0$.

\subsection{Effect of treatment allocation on variance}\label{subsec:effect_of_allocation_on_var}

In Appendix \ref{app:figures} Figures \ref{fig:vary_alloc_small_market}-\ref{fig:vary_alloc_large_market}, we observe that in a homogeneous market when $\conprob$ and $\tildeconprob$ are close to each other, the 
variance of the LR estimator is nearly symmetric about $a_L$ and $1-a_L$ and convex. Thus it is minimized at  $a_L\approx 0.5$. This observation agrees with the intuition that the choice of $a_L\approx 0.5$ balances the variances in the estimates of booking rates for the treatment and control groups. Similarly for $\CR$ experiments, we again notice that when the difference between $\conrate$ and $\tildeconrate$ is small, the $\CR$ variance is nearly symmetric about $a_C$ and $1-a_C$, and so the variance of the estimator will be minimized with $a_C\approx 0.5$.

In the case that the treatment effect is more pronounced, then the variance-optimal choice of treatment allocation may deviate from $0.5$, intuitively to balance the variance between the control estimate and the now increased treatment estimate. In most practical scenarios, however, the difference is typically small and $0.5$ should remain a reasonably near-variance-optimal choice. In
Figures \ref{fig:cr_var_approx_ratio} and \ref{fig:lr_var_approx_ratio} in Appendix \ref{app:figures}, we show the approximation ratio with the $0.5$ treatment allocation compared with the variance-optimal allocation in $\CR$ and $\LR$ designs, respectively, under different market conditions. In both designs, the variance associated with $50\%$ treatment allocation is close to the minimal variance under optimal allocation for most practical values of the $\GTE$ and market balance.

\subsection{Bias-variance tradeoff in treatment allocation}

\begin{figure}
    \centering
    \includegraphics[scale=0.5]{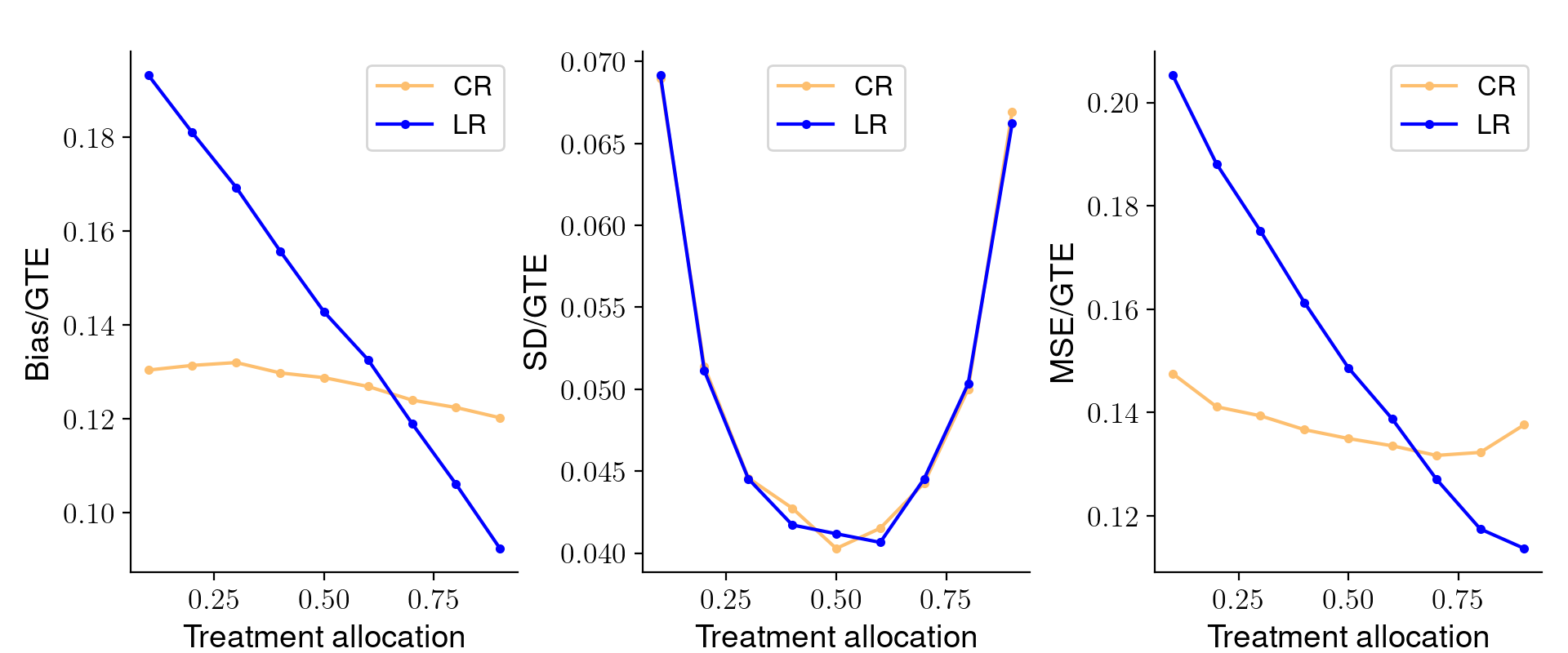}
    \caption{
    Change in bias, SD, and MSE as allocation varies. Market consists of $N=M=10^6$ homogeneous listings and customers, with parameters calibrated so that booking rate is $20\%$ in global control and $22\%$ in global treatment ($\conrate=N\conprob\approx 0.254$ and $\tildeconrate=N\tildeconprob\approx 0.286$). For each allocation ratio, bias, SD, and MSE are computed over over $3,000$ runs.}%
    \label{fig:vary_alloc_lr_cr_bias_sd}
\end{figure}

In Section \ref{sec:optimizing_exp_type}, we observe that the choice of experiment type ($\CR$ or $\LR$) does not introduce a meaningful bias-variance tradeoff; however, here show that the choice of treatment allocation does.  In particular, for both $\CR$ and $\LR$, we previously found (in a homogeneous market) that the variance minimizing allocation lies near $0.5$, but that more extreme allocations may help to minimize bias (see Figure \ref{fig:vary_alloc_lr_cr_bias_sd}). 
For a platform aiming to minimize $\MSE$, the question of whether to run a more extreme allocation to treatment and control (e.g. 75 percent in treatment and 25 percent in control) instead of a 50-50 allocation will depend on the magnitudes of the increase in standard deviation and the decrease in bias when moving to the more extreme allocation. 

\begin{figure}
    \centering
    \includegraphics[scale=0.5]{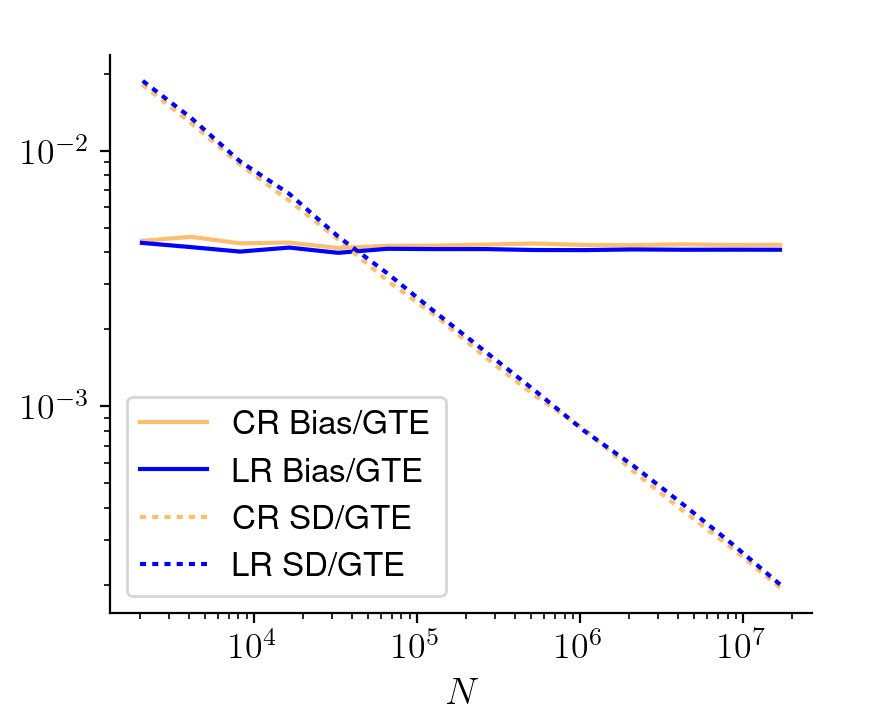}
    \caption{ The bias and SD, normalized by the $\GTE$, in balanced markets of growing sizes ranging from $2^{11}\approx 2$ thousand to $2^{24}\approx 17$ million, where both listings and customers are homogeneous. The consideration probabilities $\conprob$ and $\tildeconprob$ is such chosen that $Q_\GC=20\%$ and $Q_\GT=22\%$, namely with $\conrate=N\conprob\approx 0.254$ and $\tildeconrate=N\tildeconprob\approx 0.286$.}
    \label{fig:vary_n_bias_sd}
\end{figure}

The variance and standard deviation in the estimators is driven by the size of the market. Figure \ref{fig:vary_n_bias_sd} fixes a 50-50 allocation and shows how, as the size of the market (parameterized by $N$) increases, the standard deviation decreases while the bias remains relatively stable. In a regime where $N$ is small enough such that the standard deviation is larger than the bias even at a 50-50 split, then a platform should optimize for variance. Otherwise, the platform may wish to tradeoff some increase in variance for a decrease in bias. Figures \ref{fig:vary_alloc_small_market} - \ref{fig:vary_alloc_large_market} show how the $\MSE$ minimizing allocation lies between the variance-minimizing allocation and the bias-minimizing allocation, and is closer to the variance-minimizing allocation in a small market and closer to the bias-optimizing allocation in a large market.

\begin{figure}
    \centering
    \includegraphics[scale=0.45]{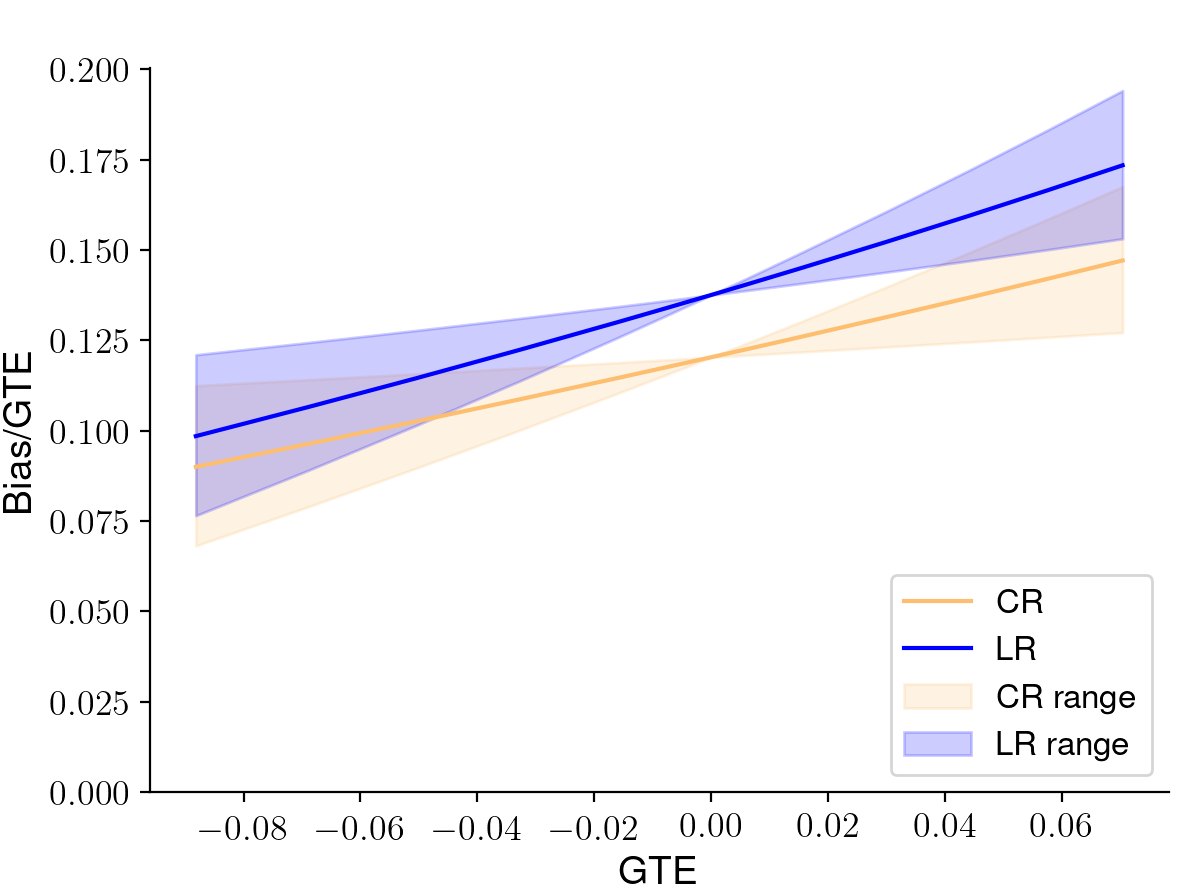}
    \caption{Numerics showing the dependence of asymptotic bias (normalized by $\GTE$) of the $\CR$ and $\LR$ estimators on the size of the treatment effect. Market is homogeneous with $\lambda=1$ in the mean field limit, with consideration probability calibrated so that $20\%$ of listings are booked under global control ($\conrate\approx 0.253$). Intervention is multiplicative with parameter $\alpha$ (i.e., $\tildeconrate = \alpha \conrate$) and $\GTE$ is varied by changing $\alpha$. Solid curves indicate the relative bias in the corresponding experiment with $50\%$ treatment. Shaded bands indicate the range of values that can be achieved through varying the allocation parameters between $a_C,
    \,a_L\in(0,1)$. 
    }
    \label{fig:lrcr_bias_vary_treatment}
\end{figure}

Of course, a platform may not know upfront the size of the bias in a given market experiment. To this end, we identify three factors that affect the bias:
\begin{enumerate}[label=\alph*)]
    \item the type of experiment run;
    \item relative demand in the market; and 
    \item the size of the treatment effect.
\end{enumerate}
Figure \ref{fig:vary_lambda_bias_sd_mse} shows the achievable $\LR$ and $\CR$ bias, standard deviation, and allocation, for treatment allocation in the range $[0.1, 0.9]$. The $\LR$ estimator is more sensitive to changes in the treatment allocation than the $\CR$ estimator.
When running an $\LR$ experiment, the bias differential between the best and worst allocation can be significant across a large range of relative demand $\lambda$. On the other hand, when running a $\CR$ experiment, the difference in bias may be relatively small for $\lambda>1$, and so a 50-50 split may be appropriate. The $\CR$ bias is somewhat more sensitive to $a_C$ for $\lambda<1$,
and so, depending on the variance, a platform may want to deviate from the 50-50 split.  Regarding the size of the treatment effect, Figure \ref{fig:lrcr_bias_vary_treatment} shows that the bias differential increases with the multiplicative lift $\alpha$ in consideration probabilities, thus also increases with the $\GTE$.

The combination of contextual knowledge and modeling can be useful in estimating the size of the bias. Many platforms may have some prior knowledge on the reasonable range for $\GTE$ and can use these bounds, along with a model calibrated to the appropriate market size and relative demand level, to estimate the size of the resulting bias. Additionally, platforms may also have estimates of bias obtained by running clustered experiments \cite{Saveski17}\cite{holtz2020reducing}. Both types of information can inform whether a platform should adjust their treatment allocation to reduce bias.

\subsection{Choice of experiment type vs. choice of treatment allocation}

While both the experiment type and treatment allocation can impact the bias in the resulting estimators, Figure \ref{fig:vary_lambda_bias_sd_mse} as well as Figures \ref{fig:vary_treat_less_demand}-\ref{fig:vary_treat_het_listings} show that, depending on the relative demand in the market, the choice of experiment type may be more important in minimizing bias. For example, in the regime where $\lambda$ is small, the estimate from a $\CR$ experiment with suboptimal allocation still has much smaller bias than the estimate from an $\LR$ experiment with optimal allocation. Likewise it is more important to run an $\LR$ experiment when $\lambda$ is large than it is to optimize the allocation in a $\CR$ experiment. However, it may be the case that many platforms exist in a scenario where supply and demand are more balanced; here the choice of allocation becomes an important lever for reducing bias. 

In the aforementioned figures, we see that in many cases the variance of the $\LR$ and $\CR$ estimators are similar (though this depends on both the relative demand and the heterogeneity on the platform). In these cases, the allocation is important for minimizing variance.

\section{Discussion and Future Work}
\label{sec:discussion}

Our findings have implications for experiment design in practice beyond solely minimizing bias and variance As one example, an important factor for platforms is the risk involved in an experiment. If the platform believes that the intervention has a chance of harming a metric of interest, then it might choose to allocate fewer individuals to the treatment group to start out.  Depending on the performance on this initial set, the platform either increases the treatment allocation (if the intervention seems promising) or stops the experiment (if key metrics are harmed). This type of experiment is often referred to as a ``ramp-up'' experiment \cite{xu2018sqr}. Though the choices made in these ramp-up experiments are generally made independently of the concerns we study about interference and bias, our work can be used to show that these decisions may happen to be optimal for reducing bias.

We illustrate these implications in a homogeneous market with one listing type and one customer type.\footnote{We conjecture that similar findings hold in heterogeneous markets, as long as the treatment effect is not too heterogeneous across types.} Our results show that the bias of $\CR$ and $\LR$ estimators are monotonic in the treatment proportion. For a $\CR$ experiment with probability of treatment $a_C$, Theorem \ref{Thm_opt_bias_ac} shows that if the $\GTE$ is positive, then the bias decreases as $a_C$ increases and if the $\GTE$ is negative, bias decreases as $a_C$ increases. For an $\LR$ experiment with probability of treatment $a_L$, Theorem \ref{thm:opt_bias_aL} shows that the bias is monotonic in $a_L$, although the direction of change depends on both the $\GTE$ and whether the relative demand $\lambda$ is less than some cutoff $\lambda^*$. However, we find that in many scenarios, as long as the number of customers is not too much larger than the number of listings, 
$\lambda$ is less than the cutoff and the $\LR$ bias behaves similarly to the $\CR$ bias; that is, the bias of $\LR$ is decreasing in $a_L$ when $\GTE>0$ and increasing in $a_L$ when $\GTE<0$.  
We first consider this case of a reasonably small $\lambda$ less than the cutoff, where both $\CR$ and $\LR$ bias are decreasing in the treatment proportion.

When a platform deems an intervention as potentially ``risky" and runs a ramp-up experiment, it is implicitly stating that there is a non-negligible chance that the $\GTE<0$. Note that when $\GTE<0$, a smaller treatment allocation actually reduces bias and helps the platform more accurately ascertain the drop in bookings. In other words, the initial allocation is beneficial for precisely the scenario that the platform is worried about.  On the other hand, suppose that in the ramp-up experiment, we actually have $\GTE>0$. This means that the initial allocation will lead us to overestimate the benefit of the intervention. However, in this case, the bias is not necessarily detrimental: upon seeing positive changes in bookings, the platform will increase the allocation to treatment and thereby decrease the bias in the estimator.  Thus when an intervention is risky and the platform chooses a ramp up experiment, the preceding discussions suggests that the adaptive sequential increase in allocation is beneficial {\em both} for measuring a negative effect if the $\GTE<0$ {\em and} for measuring a positive effect if $\GTE>0$. 

However, the cautionary note is that if the platform is running an $\LR$ experiment and $\lambda$ is sufficiently large, then the sequential increase will have the opposite effect. The initial allocation will overestimate the effect of a detrimental intervention if $\GTE<0$. If $\GTE>0$, then the final allocations with an increased proportion of listings randomized to treatment will lead to a greater overestimate of the $\GTE$. 

Our work introduces a model through which many practical designs and considerations can be studied. The model captures marketplace competition effects and interference, and yet is simple enough that the bias and variance can be fully characterized. Future directions of study include a richer class of estimators, beyond the standard $\CR$ and $\LR$ difference-in-means estimators studied here. Additionally, the model can be used to study other experimental designs that,  for example, randomize at both sides of the market simultaneously  \cite{bajari2019double, johari2021experimental} or randomize on clusters of individuals \cite{chamandy16, holtz2020reducing}. Finally, we hope that the tractability of our model can shed light on the joint optimization of design and analysis in the context of marketplace experiments.

\section*{Acknowledgement}
	This work was supported by the National Science Foundation under grants 1931696 and 1839229 and the Dantzig-Lieberman Operations Research Fellowship.

\bibliographystyle{abbrv}
\bibliography{bibliography}

\appendix

\section{Proofs}
\label{Appendix:proofs}

\lemmaLimAppBookProb*
\begin{proof}
This result is essentially a consequence of the well-known fact that a sequence of Poisson binomial distributions (with appropriate parameters) converges to a Poisson distribution in total variance distance. In our lemma, since our primary focus is the expected number of bookings, we will present the following argument targeting the in expectation limit.

We will first prove the convergence \eqref{Eqn_limit_app_prob:Prop_limit_app_and_book_prob} for the application step, namely
\begin{equation}\label{Eqn_lim_app_prob_restate}
    N \appprob(\gamma,\theta) \to \conrate(\gamma,\theta) \frac{1-\exp(-\v{\conratematrix}(\gamma,\cdot)\cdot \bm{\tau})}{\v{\conratematrix}(\gamma,\cdot)\cdot \bm{\tau}}
\end{equation}
based on the convergence $N \conprob(\gamma,\theta) \to \conrate(\gamma,\theta)$ for all $\ginG$ and $\tinT$.
The convergence \eqref{Eqn_limit_book_prob:Prop_limit_app_and_book_prob} for the acceptance step can in turn be shown analogously from the convergence of $N\appprob(\gamma,\theta)$.

To show \eqref{Eqn_lim_app_prob_restate}, we first notice that by linearity of expectation
\[
    \mathbb{P}(\text{customer } c \text{ makes an application}) = \sum_{l=1}^n \appprob(\gamma_c,\theta_l) = \sum_{\tinT} t(\theta) \appprob(\gamma_c,\theta).
\]
Our model postulates that customer $c$ will make an application if and only if their consideration set $S_c$ is non-empty, or equivalently $\sum_{\tinT} C_{c,\theta}>0$, where random variable $C_{c,\theta}$ denotes the number of type-$\theta$ listings included in $\mathcal{S}_c$ and hence $C_{c,\theta}\sim\Binom(t(\theta), \conprob(\gamma_c,\theta))$. Thus,
\begin{equation}\label{Eqn_conv_sum_p_app:Prop_limit_app_and_book_prob}
    \sum_{\tinT} t(\theta) \appprob(\gamma_c,\theta) = \prod_{\tinT} \mathbb{P}(C_{c,\theta}=0) = \prod_{\tinT} (1-\conprob(\gamma_c,\theta))^{t(\theta)} \to 1-\exp(-\v{\conratematrix}(\gamma,\cdot)\cdot \bm{\tau})
\end{equation}
as $N\to\infty$, since $N\conprob(\gamma,\theta)\to \conrate(\gamma,\theta)$ and $t(\theta)/N \to \tau(\theta)$.

Now it suffices for us to show the convergence of the ratios between $\appprob(\gamma_i,\theta)$ and $\appprob(\gamma_i,\theta')$ for any two listing types.
Consider the following alternative criterion for each customer to decide which listing to apply to. Fix an arbitrary customer $c$ of type $\gamma$, and for each listing $l$ generate independently a random score $X_{cl} \sim \Unif(0, \conprob(\gamma,\theta_l)^{-1})$. A listing $l$ is considered available to $c$ if $X_{cl} < 1$, and among all the available listings (if any), the customer $c$ will apply to the one with minimal score.\footnote{The event of multiple listings having equal score has probability zero and hence tie breaking does not affect our analysis.} It is straightforward to verify that this procedure yields a result consistent with the original framework. Further, this now allows us to express $\appprob(\gamma,\theta)$ conveniently in an integral form, namely
\begin{align}
    \appprob(\gamma,\theta) &= \mathbb{P}(\text{customer } c \text{ applies to listing } l \text{ of type } \theta) \nonumber \\
        &= \mathbb{P}(X_{cl}=\min_{l'} X_{cl'}\text{ and }X_{cl}<1) \nonumber \\
        &= \int_0^1 \mathbb{P}(X_{cl'}>x \;\forall\, l'\neq l \;|\;X_{cl}=x)f_{X_{cl}}(x)dx \nonumber \\
        &= \conprob(\gamma,\theta) \int_0^1 (1-x\conprob(\gamma,\theta))^{t(\theta)-1}\prod_{\substack{\theta'\in\Theta\\\theta'\neq\theta}} (1-x\conprob(\gamma,\theta'))^{t(\theta')-1} dx. \label{Eqn_p_app_exact_integral_form:Prop_limit_app_and_book_prob}
\end{align}
For two listing types $\theta_1,\theta_2\in\Theta$, the ratio of their application probabilities (to a single listing of the corresponding type) is
\[
    \frac{\appprob(\gamma,\theta_1)}{\appprob(\gamma,\theta_2)} = \frac{\conprob(\gamma,\theta_1)}{\conprob(\gamma,\theta_2)} \cdot \frac{\int_0^1 (1-x\conprob(\gamma,\theta_1))^{t(\theta_1)-1}\prod_{\substack{\theta'\in\Theta\\\theta'\neq\theta_1}} (1-x\conprob(\gamma,\theta'))^{t(\theta')-1} dx}{\int_0^1 (1-x\conprob(\gamma,\theta_2))^{t(\theta_2)-1}\prod_{\substack{\theta'\in\Theta\\\theta'\neq\theta_2}} (1-x\conprob(\gamma,\theta'))^{t(\theta')-1} dx}.
\]
Focusing on the ratio of the two integrals, we notice that the ratio of the integrands is simply
\begin{equation}\label{Eqn_ratio_p_app_theta1_theta2:Prop_limit_app_and_book_prob}
    \frac{1-x\conprob(\gamma,\theta_2)}{1-x\conprob(\gamma,\theta_1)}
\end{equation}
for $x\in (0,1)$. Recall that as $N\to\infty$, $N\conprob^{(N)}(\gamma,\theta)$ converges to some constant for each pair of $\gamma$ and $\theta$. Thus, $\conprob(\gamma,\theta_1)$ and $\conprob(\gamma,\theta_2)$ both approaches 0 as $N\to\infty$. As a result, the integrand ratio \eqref{Eqn_ratio_p_app_theta1_theta2:Prop_limit_app_and_book_prob} converges to 1 uniformly on the domain, and hence the ratio between the integrals also converges to 1. In other words,
\begin{equation}\label{Eqn_conv_ratio_p_app:Prop_limit_app_and_book_prob}
    \frac{\appprob(\gamma,\theta_1)}{\appprob(\gamma,\theta_2)} \to \frac{\conrate(\gamma,\theta_1)}{\conrate(\gamma,\theta_2)}.
\end{equation}
The fact that \eqref{Eqn_conv_ratio_p_app:Prop_limit_app_and_book_prob} holds for any pair of listing types combined with \eqref{Eqn_conv_sum_p_app:Prop_limit_app_and_book_prob} implies that
\[
    t(\theta) \appprob(\gamma,\theta) \to \frac{\tau(\theta)\conrate(\gamma,\theta)}{\v{\conratematrix}(\gamma,\cdot)\cdot \bm{\tau}} \left(1-\exp(-\v{\conratematrix}(\gamma,\cdot)\cdot \bm{\tau})\right),
\]
which immediately implies that $N\appprob(\gamma,\theta) \to \apprate(\gamma,\theta)$ as we claimed.

For the acceptance step, the limit
\[
    \bookprob(\gamma,\theta) \to \apprate(\gamma,\theta) \frac{1-\exp(-\lambda\bsigma\cdot\v{\appratematrix}(\cdot,\theta))}{\bsigma\cdot\v{\appratematrix}(\cdot,\theta)}
\]
as $N\to\infty$ can be established in the exact same way as above, this time in turn using the fact that $N\appprob(\gamma,\theta) \to \apprate(\gamma,\theta)$ that we proved just now.
\end{proof}

\corLimBookingRate*
\begin{proof}
By linearity of expectation
\[
    \E [Q] = \sum_{i=1}^M\sum_{j=1}^N \bookprob(\gamma_i,\theta_j) = \sum_{\ginG,\tinT} s(\gamma) t(\theta) \bookprob(\gamma,\theta).
\]
By Proposition~\ref{Prop_limit_app_and_book_prob}, we have $N\bookprob(\theta,\gamma)\to \bookrate(\theta,\gamma)$. Recall that we assume $N^{-1}t(\theta)\to\tau(\theta)$ and $M^{-1}s(\gamma)\to\sigma(\gamma)$ for each $\ginG$ and $\tinT$, and $M/N\to\lambda$ as $N\to\infty$. Thus,
\begin{align*}
    N^{-1} \E \left[Q\right] &= N^{-1} \sum_{\ginG,\tinT} s(\gamma) t(\theta) \bookprob(\gamma,\theta) \\
    &= \frac{M}{N} \sum_{\ginG,\tinT} \frac{s(\gamma)}{M} \frac{t(\theta)}{N} \cdot N \bookprob(\gamma,\theta) \\
    &\to \lambda \sum_{\ginG,\tinT} \sigma(\gamma) \tau(\theta) \bookrate(\gamma,\theta).
\end{align*} \end{proof}

\begin{restatable}[Limits of expectations of $\CR$ and $\LR$ estimators]{proposition}{propLimEstimatorExpectation}
\label{Prop_limit_estimator_expectation}
Along the sequence as $N\to\infty$, the naive estimators for $\CR$ and $\LR$ converge in expectation
to the following limits
\begin{align}\label{Eqn_limit_cr_est_expectation:Prop_limit_estimator_expectation}
    \E \gtehatcr^{(N)}(a_C) &\to \lambda \sum_{\tinT} \tau(\theta) \bsigma \cdot \left(\v{\tildeappratematrix}(\cdot,\theta) - \v{\appratematrix}(\cdot,\theta)\right) F\left(\lambda \bsigma \cdot \big(a_C \v{\tildeappratematrix}(\cdot,\theta) + (1-a_C) \v{\appratematrix}(\cdot,\theta)\big)\right) \nonumber\\
    &= \lambda \sum_{\tinT} \tau(\theta) \left(\sum_{\ginG}\sigma(\gamma) \left(\tildeconrate(\gamma,\theta) F(\v{\tau}\cdot\v{\tildeconratematrix}(\gamma,\cdot)) - \conrate(\gamma,\theta) F(\v{\tau}\cdot\v{\conratematrix}(\gamma,\cdot))\right)\right) \nonumber\\
    &\qquad F\left(\lambda \sum_{\ginG}\sigma(\gamma) \left(a_C \tildeconrate(\gamma,\theta) F(\v{\tau}\cdot\v{\tildeconratematrix}(\gamma,\cdot)) + (1-a_C) \conrate(\gamma,\theta) F(\v{\tau}\cdot\v{\conratematrix}(\gamma,\cdot))\right)\right),
\end{align}

\begin{multline}\label{Eqn_limit_lr_est_expectation:Prop_limit_estimator_expectation}
    \E \gtehatlr^{(N)}(a_L) \to \sum_{\tinT} \tau(\theta) \left(
    \exp\bigg(
    -\lambda \sum_{\ginG}\sigma(\gamma) \conrate(\gamma,\theta)
    F\left(\v{\tau}\cdot\big(a_L\v{\tildeconratematrix}(\gamma,\cdot) + (1-a_L)\v{\conratematrix}(\gamma,\cdot)\big)\right)
    \bigg)\right.
    \\
    \left. -
    \exp\bigg(
    -\lambda \sum_{\ginG}\sigma(\gamma) \tildeconrate(\gamma,\theta)
    F\left(\v{\tau}\cdot\big(a_L\v{\tildeconratematrix}(\gamma,\cdot) + (1-a_L)\v{\conratematrix}(\gamma,\cdot)\big)\right)
    \bigg)
    \right).
\end{multline}
\end{restatable}

\begin{proof}
These expressions are a direct consequence of Proposition~\ref{Prop_limit_app_and_book_prob} when we consider types of customer and listing in the product spaces $\Gamma\times\{0,1\}$ and $\Theta\times\{0,1\}$.

Let us first consider the treatment booking rate in $\LR$ experiments. In this completely randomized design, we know that the fractional size $N^{-1}t(\theta,1)$ of the treated portion of type-$\theta$ listings converges to $a_L\tau(\theta)$ for all listing types with probability 1. Here the underlying probability space is all possible realizations of how treatment is assigned to the listings of different types in the sequence of markets indexed by $N$. This now allows us to condition on each realization of treatment assignment, and apply Proposition~\ref{Prop_limit_app_and_book_prob} using $\Theta\times\{0,1\}$ as the listing type space. We then have the following convergence for the probability for a customer of type $\gamma$ to book a treated type-$\theta$ listing
\begin{equation}\label{Eqn_p_book_treatment:Prop_limit_estimator_expectation}
    N\bookprob(\gamma,(\theta,1)) \to \apprate(\gamma,(\theta,1)) \frac{1-\exp(-\lambda\bsigma\cdot\v{\appratematrix}(\cdot,(\theta,1)))}{\bsigma\cdot\v{\appratematrix}(\cdot,(\theta,1))},
\end{equation}
where $\apprate(\gamma,(\theta,1))$ is the application rate to treated type-$\theta$ listings defined in analogy to \eqref{Eqn_redef_app_rate_using_F} as
\[
    \apprate(\gamma,(\theta,1)) = \tildeconrate(\gamma,\theta) F\left(\v{\tau}\cdot\left((1-a_C)\v{\conratematrix}(\gamma,\cdot) + a_C\v{\tildeconratematrix}(\gamma,\cdot)\right)\right).
\]
Substituting $\apprate(\gamma,(\theta,1))$ into \eqref{Eqn_p_book_treatment:Prop_limit_estimator_expectation} and summing over all customers and treated listings, we obtain the following through basic arithmetic operations
\begin{align}
    \E\left[Q_{\LR}^{(N)}(1|a_L)\middle|t(\theta,1),\,\tinT\right]
    &= \frac{M}{N} \sum_{\ginG,\tinT} \frac{t(\theta,1)}{N}\frac{s(\gamma)}{M}\cdot N \bookprob(\gamma,(\theta,1)) \nonumber\\
    &\to \lambda\sum_{\ginG,\tinT} a_L\tau(\theta)\sigma(\gamma)
    \apprate(\gamma,(\theta,1)) \frac{1-\exp(-\lambda\bsigma\cdot\v{\appratematrix}(\cdot,(\theta,1)))}{\bsigma\cdot\v{\appratematrix}(\cdot,(\theta,1))} \nonumber\\
    &= a_L \sum_{\tinT} \tau(\theta)\left(1-\exp(-\lambda\bsigma\cdot\v{\appratematrix}(\cdot,(\theta,1)))\right)\nonumber\\
    &= a_L - a_L\sum_{\tinT} \tau(\theta) \nonumber\\
    &\qquad \exp\left(
    -\lambda \sum_{\ginG}\sigma(\gamma) \tildeconrate(\gamma,\theta) F\left(\v{\tau}\cdot\left(a_L\v{\tildeconratematrix}(\gamma,\cdot) + (1-a_L)\v{\conratematrix}(\gamma,\cdot)\right)\right)
    \right). \label{Eqn_limit_lr_q1_est_expectation:Prop_limit_estimator_expectation}
\end{align}
Note that the expectation above is taken over the random realization of the consideration, application, and booking steps, conditioned on the sequence of treatment assignments. However, the treatment assignment satisfies $t(\theta,1)/N\to a_L\tau(\theta)$ for all $\theta\in\Theta$, which occurs with probability 1 by the strong law of large numbers. By the boundedness of $N^{-1}Q_{\LR}^{(N)}(1|a_L)$, we conclude that \eqref{Eqn_limit_lr_q1_est_expectation:Prop_limit_estimator_expectation} indeed holds when we include the random treatment assignments into our probability space.

Similarly, for the control listings in $\LR$ experiments, we have
\begin{multline}\label{Eqn_limit_lr_q0_est_expectation:Prop_limit_estimator_expectation}
    \frac{1}{(1-a_L)N}\E\left[Q_{\LR}^{(N)}(0|a_L)\right] \to 1 - \sum_{\tinT} \tau(\theta) \\
    \exp\left(
    -\lambda \sum_{\ginG}\sigma(\gamma) \conrate(\gamma,\theta) F\left(\v{\tau}\cdot\left(a_L\v{\tildeconratematrix}(\gamma,\cdot) + (1-a_L)\v{\conratematrix}(\gamma,\cdot)\right)\right)
    \right).
\end{multline}
Combining \eqref{Eqn_limit_lr_q1_est_expectation:Prop_limit_estimator_expectation} with \eqref{Eqn_limit_lr_q0_est_expectation:Prop_limit_estimator_expectation} gives our claimed limit of \eqref{Eqn_limit_lr_est_expectation:Prop_limit_estimator_expectation}.

For $\CR$ experiments, the proof of \eqref{Eqn_limit_cr_est_expectation:Prop_limit_estimator_expectation} goes analogously. Again, by Proposition~\ref{Prop_limit_app_and_book_prob}, we have
\begin{equation}\label{Eqn_limit_cr_q1_est_expectation:Prop_limit_estimator_expectation}
    \frac{1}{a_C N}Q_{\CR}^{(N)}(1|a_C) \to \lambda \sum_{\tinT} \tau(\theta) \bsigma \cdot \tildeappratematrix(\cdot,\theta) F\left(\lambda \bsigma \cdot \left(a_C \v{\tildeappratematrix}(\cdot,\theta) + (1-a_C) \v{\appratematrix}(\cdot,\theta)\right)\right),
\end{equation}
and
\begin{equation}\label{Eqn_limit_cr_q0_est_expectation:Prop_limit_estimator_expectation}
    \frac{1}{(1-a_C)N}Q_{\CR}^{(N)}(0|a_C) \to \lambda \sum_{\tinT} \tau(\theta) \bsigma \cdot \v{\appratematrix}(\cdot,\theta) F\left(\lambda \bsigma \cdot \left(a_C \v{\tildeappratematrix}(\cdot,\theta) + (1-a_C) \v{\appratematrix}(\cdot,\theta)\right)\right),
\end{equation}
where $\apprate$ and $\tildeapprate$ take the exact same form as defined in \eqref{Eqn_redef_app_rate_using_F}. Combining \eqref{Eqn_limit_cr_q1_est_expectation:Prop_limit_estimator_expectation} and \eqref{Eqn_limit_cr_q0_est_expectation:Prop_limit_estimator_expectation} gives the limit \eqref{Eqn_limit_cr_est_expectation:Prop_limit_estimator_expectation}.
\end{proof}

\lemmaUpwardBias*

\begin{proof}
By Corollary~\ref{Corollary_limit_booking_rate}, the limit booking rate under global control is
\begin{align}
    N^{-1}Q_{GC}^{(N)} \to \lambda \bm{\sigma}^T \bookratematrix \bm{\tau}
    &= \sum_{\tinT} \tau(\theta) \left(1-\exp(-\lambda\bsigma\cdot\v{\appratematrix}(\cdot,\theta))\right) \nonumber\\
    &= 1 - \sum_{\tinT} \tau(\theta) \exp\left(-\lambda \sum_{\ginG}\sigma(\gamma) \conrate(\gamma,\theta) F(\v{\tau}\cdot\v{\conratematrix}(\gamma,\cdot))\right) \nonumber
\end{align}
and the similar convergence applies to the limit booking rate under global treatment.
Compare this with the expression \eqref{Eqn_limit_lr_est_expectation:Prop_limit_estimator_expectation}, we notice that the difference is only in the argument inside $F(\cdot)$, where
\[
    \v{\tau}\cdot\v{\conratematrix}(\gamma,\cdot) \leq \v{\tau}\cdot\left(a_L\v{\tildeconratematrix}(\gamma,\cdot) + (1-a_L)\v{\conratematrix}(\gamma,\cdot)\right) \leq \v{\tau}\cdot\v{\tildeconratematrix}(\gamma,\cdot).
\]
Since $F$ is monotone decreasing and the function $x\mapsto 1-e^{-x}$ is monotone increasing, we find
\begin{equation}\label{Eqn_comp_qgc_qlr0:Lemma_upward_bias}
    \lim_{N\to\infty} N^{-1}Q_{GC}^{(N)} \geq \lim_{N\to\infty} N_0^{-1}Q_{\LR}^{(N)}(0|a_L),
\end{equation}
and similarly one can show
\begin{equation}\label{Eqn_comp_qgt_qlr1:Lemma_upward_bias}
    \lim_{N\to\infty} N^{-1}Q_{GT}^{(N)} \leq \lim_{N\to\infty} N_1^{-1}Q_{\LR}^{(N)}(1|a_L).
\end{equation}
Combining \eqref{Eqn_comp_qgc_qlr0:Lemma_upward_bias} and \eqref{Eqn_comp_qgt_qlr1:Lemma_upward_bias} gives the asymptotic positivity of $\LR$ bias.

To consider the limit expectation of the $\CR$ estimator, we may write the limit booking rate under global control as
\[
    \lambda \sum_{\tinT} \tau(\theta) \bsigma \cdot \v{\appratematrix}(\cdot,\theta) F\left(\lambda \bsigma \cdot \v{\appratematrix}(\cdot,\theta)\right).
\]
It is easy to verify from \eqref{Eqn_def_phi_app:Prop_limit_app_and_book_prob} and the multiplicative assumption that
\begin{align}
    \tildeapprate(\gamma,\theta) &= \tildeconrate(\gamma,\theta) \frac{1-\exp(-\v{\tau}\cdot\v{\tildeconratematrix}(\gamma,\cdot))}{\v{\tau}\cdot\v{\tildeconratematrix}(\gamma,\cdot)} \nonumber\\
    &= \alpha\conrate(\gamma,\theta) \frac{1-\exp(-\v{\tau}\cdot\alpha\v{\conratematrix}(\gamma,\cdot))}{\v{\tau}\cdot\alpha\v{\conratematrix}(\gamma,\cdot)} \nonumber\\
    &\geq \conrate(\gamma,\theta) \frac{1-\exp(-\v{\tau}\cdot\v{\conratematrix}(\gamma,\cdot))}{\v{\tau}\cdot\v{\conratematrix}(\gamma,\cdot)} = \apprate(\gamma,\theta) \nonumber
\end{align}
since $\alpha > 1$. That is $\tildeappratematrix > \appratematrix$ component-wise. Thus,
\[
    F\left(\lambda \bsigma \cdot \v{\appratematrix}(\cdot,\theta)\right) \geq F\left(\lambda \bsigma \cdot \left(a_C \v{\tildeappratematrix}(\cdot,\theta) + (1-a_C) \v{\appratematrix}(\cdot,\theta)\right)\right)
\]
and
\[
    \lim_{N\to\infty} N^{-1}Q_{GC}^{(N)} \geq \lim_{N\to\infty} \frac{1}{(1-a_C)N}Q_{\CR}^{(N)}(0|a_C).
\]
In the same way, we have
\[
    \lim_{N\to\infty} N^{-1}Q_{GT}^{(N)} \leq \lim_{N\to\infty} \frac{1}{a_CN}Q_{\CR}^{(N)}(1|a_C),
\]
completing our proof of the asymptotic positivity of $\CR$ bias.
\end{proof}

\thrmUnbiasedExtremes*

\begin{proof}
We denote the limit expectations of the naive $\CR$ and $\LR$ estimators by
\[
E_\CR(a_C) = \lim_{N\to\infty}N^{-1}\gtehatcr(a_C) \;\text{ and }\; E_\LR(a_L) = \lim_{N\to\infty}N^{-1}\gtehatlr(a_L).
\]
We will derive the results in the market extremes from our expressions for the $E_\CR$ and $E_\LR$ given by \eqref{Eqn_limit_cr_est_expectation:Prop_limit_estimator_expectation} and \eqref{Eqn_limit_lr_est_expectation:Prop_limit_estimator_expectation}.

First consider the limit with $\lambda\to 0$. In this case, since the customer side limits the number of bookings, we focus on the treatment effect measured from the customer side, i.e. $(\lambda N)^{-1}\widehat{\GTE}$. From \eqref{Eqn_limit_cr_est_expectation:Prop_limit_estimator_expectation}, we have
\[
    \lim_{\lambda\to 0} \frac{E_\CR(a_C)}{\lambda} = \sum_{\tinT} \tau(\theta) \bsigma \cdot \left(\v{\tildeappratematrix}(\cdot,\theta) - \v{\appratematrix}(\cdot,\theta)\right)
\]
as $\lim_{\lambda\to 0} F\left(\lambda \bsigma \cdot \left(a_C \v{\tildeappratematrix}(\cdot,\theta) + (1-a_C) \v{\appratematrix}(\cdot,\theta)\right)\right) = 0$ for any $\tinT$, and by L'Hospital's rule
\[
    \lim_{\lambda\to 0} \frac{E_\LR(a_L)}{\lambda} = \sum_{\tinT,\ginG} \tau(\theta)
    \sigma(\gamma) (\tildeconrate(\gamma,\theta) - \conrate(\gamma,\theta))
    F\left(\v{\tau}\cdot\left(a_L\v{\tildeconratematrix}(\gamma,\cdot) + (1-a_L)\v{\conratematrix}(\gamma,\cdot)\right)\right).
\]
By Corollary~\ref{Corollary_limit_gte}, for the true $\GTE$,
\begin{align*}
    \lim_{N\to\infty} \frac{\GTE}{\lambda N} &= \sum_{\tinT}  \tau(\theta) \lambda^{-1} \left( \exp(-\lambda\bsigma\cdot\v{\appratematrix}(\cdot,\theta)) - \exp(-\lambda\bsigma\cdot\v{\tildeappratematrix}(\cdot,\theta)) \right) \\
    &\to \sum_{\tinT} \tau(\theta) \bsigma \cdot \left(\v{\tildeappratematrix}(\cdot,\theta) - \v{\appratematrix}(\cdot,\theta)\right)
\end{align*}
as $\lambda\to 0$,
which verifies that the naive $\CR$ estimator is asymptotically unbiased in the supply-constrained market extreme when the naive $\LR$ estimator is not.

For the market extreme as $\lambda\to\infty$, we similarly compute the limit of GTE
\[
    \lim_{N\to\infty} \frac{\GTE}{ N} = \sum_{\tinT} \tau(\theta)  \left( \exp(-\lambda\bsigma\cdot\v{\appratematrix}(\cdot,\theta)) - \exp(-\lambda\bsigma\cdot\v{\tildeappratematrix}(\cdot,\theta)) \right) \to 0
\]
as $\lambda\to\infty$. Taking limit of \eqref{Eqn_limit_lr_est_expectation:Prop_limit_estimator_expectation} as $\lambda\to\infty$ reveals that for $\LR$ we also have $\lim_{\lambda\to \infty} E_\LR(a_L) = 0$. For the $\CR$ case, however,
\[
    \lim_{\lambda\to 0} E_\CR(a_C) = \sum_{\tinT} \tau(\theta) \frac{ \bsigma \cdot \left(\v{\tildeappratematrix}(\cdot,\theta) - \v{\appratematrix}(\cdot,\theta)\right)}{\bsigma \cdot \left(a_C \v{\tildeappratematrix}(\cdot,\theta) + (1-a_C) \v{\appratematrix}(\cdot,\theta)\right)},
\]
which is in general nonzero.
Thus, the naive $\LR$ estimator is asymptotically unbiased in the demand-constrained market extreme while the naive $\CR$ estimator is not.
\end{proof}

\thrmCROptBiasAC*

\begin{proof}
We prove the result in the case where $\alpha>1$. The case where $\alpha<1$ holds by symmetry.
By Lemma~\ref{Lemma_upward_bias}, we know the asymptotic bias of the $\CR$ estimator is positive for any fixed choice of $a_C\in(0,1)$ when the intervention is positive and multiplicative on consideration probability. Thus, minimizing the asymptotic bias is equivalent to minimizing the limit of the expectation of $\gtehatcr$, given in Proposition~\ref{Prop_limit_estimator_expectation} as
\[
    \lim_{N \rightarrow \infty}  \E\left[\widehat{\GTE}_\CR^{(N)}(a_C)\right] = \lambda \sum_{\tinT} \tau(\theta) \bsigma \cdot \left(\v{\tildeappratematrix}(\cdot,\theta) - \v{\appratematrix}(\cdot,\theta)\right) F\left(\lambda \bsigma \cdot \left(a_C \v{\tildeappratematrix}(\cdot,\theta) + (1-a_C) \v{\appratematrix}(\cdot,\theta)\right)\right).
\]
Taking partial derivative with respect to $a_C$, we find
\begin{equation}\label{Eqn_partial_derivative_ac:Thm_opt_bias_ac}
    \frac{\partial}{\partial a_C} B_{\CR}(a_C) = \lambda^2 \sum_{\tinT} \tau(\theta) \left(\bsigma \cdot \left(\v{\tildeappratematrix}(\cdot,\theta) - \v{\appratematrix}(\cdot,\theta)\right)\right)^2 F'\left(A\right),
\end{equation}
where $A$ is the short-hand for the argument in $F(\cdot)$ above. This value is strictly negative since $F'(x) < 0$ for any $x > 0$. Thus, $B_{\CR}(a_C)$ is positive strictly decreasing on $(0,1)$.
\end{proof}

\thrmLROptBiasAL*

\begin{proof}
We assume that $\tildeconprob > \conprob$ and hence $\tildeconrate > \conrate$ in the limit. The opposite case holds by symmetry.

Taking partial derivative of the limit in \eqref{Eqn_homogen_limit_lr_est_expectation:Corollary_homogen_limit_estimator_expectation} with respect to $a_L$ gives
\[
    \frac{\partial}{\partial a_L}  B_\LR(a_L) = -\tildeconrate \exp\left(-\lambda \tildeconrate F(a_L\tildeconrate + (1-a_L)\conrate)\right) + \conrate \exp\left(-\lambda \conrate F(a_L\tildeconrate + (1-a_L)\conrate)\right).
\]
This means that
\[
    \frac{\partial}{\partial a_L}  B_\LR(a_L) > 0 \;\Longleftrightarrow\; \frac{\tildeconrate}{\conrate} < \exp\left(\lambda(\tildeconrate-\conrate) F(a_L\tildeconrate + (1-a_L)\conrate)\right).
\]
Notice that the left hand side is constant, and the right hand side increases in $\lambda$ and decreases in $a_L$. Hence the function $B_\LR(a_L)$ admits at most one local maximum in the interior of $(0,1)$. This implies that the infimum of $B_\LR$ must be achieved on the boundary of the interval. We can verify the statement on $\lambda^*$ through an algebraic comparison of $B_\LR(0)$ and $B_\LR(1)$, both can be defined by continuity.
\end{proof}
\section{Additional Simulations}
\label{app:figures}

In this section, we compare the bias, standard deviation, and $\MSE$ of the $\CR$ and $\LR$ estimators in different market settings. 

In Figures \ref{fig:vary_treat_less_demand} - \ref{fig:vary_treat_more_demand}, we show the behavior of the estimators as we vary the multiplicative treatment lift $\alpha$ (and hence varying the $\GTE$). Each figure shows this behavior at a different level of relative demand. All three figures depict a homogeneous market.

\begin{figure}[H]
    \centering
    \includegraphics[scale=0.5]{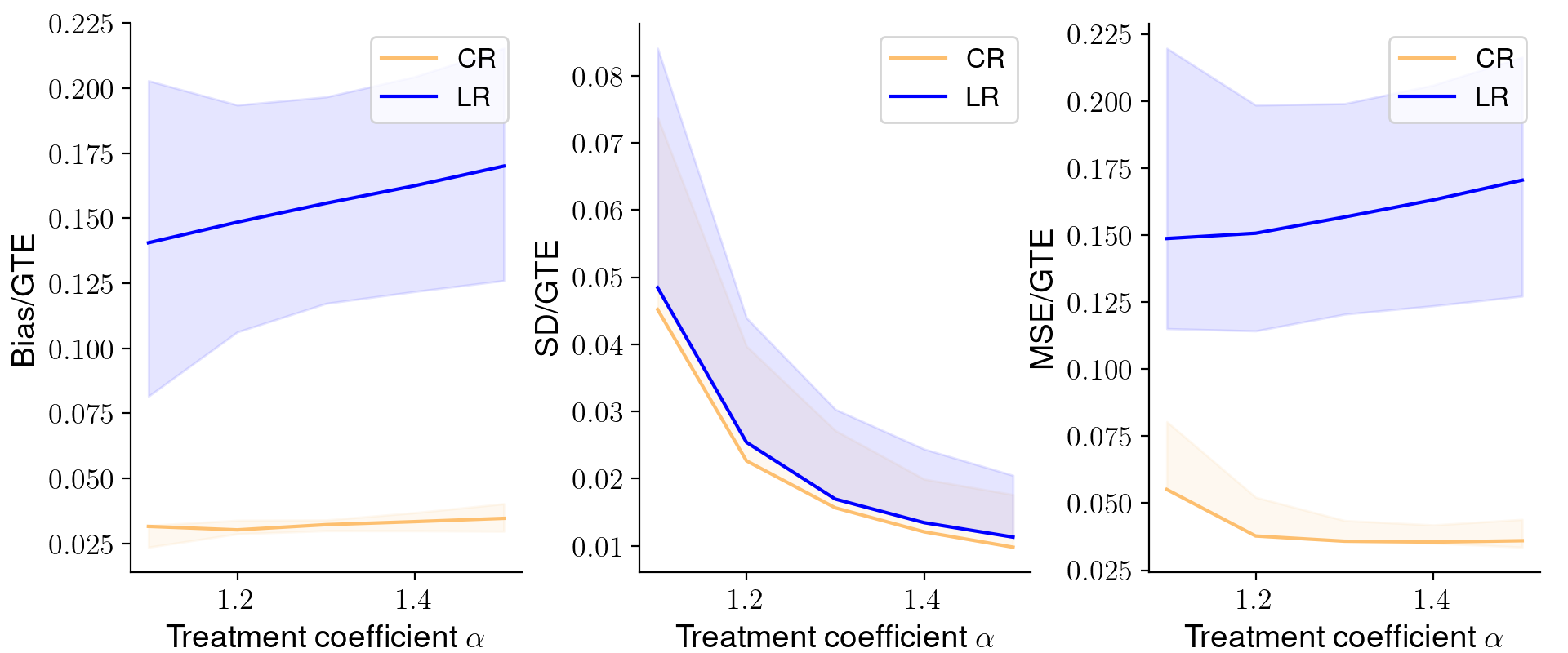}
    \caption{Demand constrained market. In this example, we consider $M=2^{20}\approx 1.0$ million homogeneous customers and $N=4M=2^{22}\approx 4.2$ million homogeneous listings, with $\conrate=N\conprob=0.253$.}
    \label{fig:vary_treat_less_demand}
\end{figure}

\begin{figure}[H]
    \centering
    \includegraphics[scale=0.5]{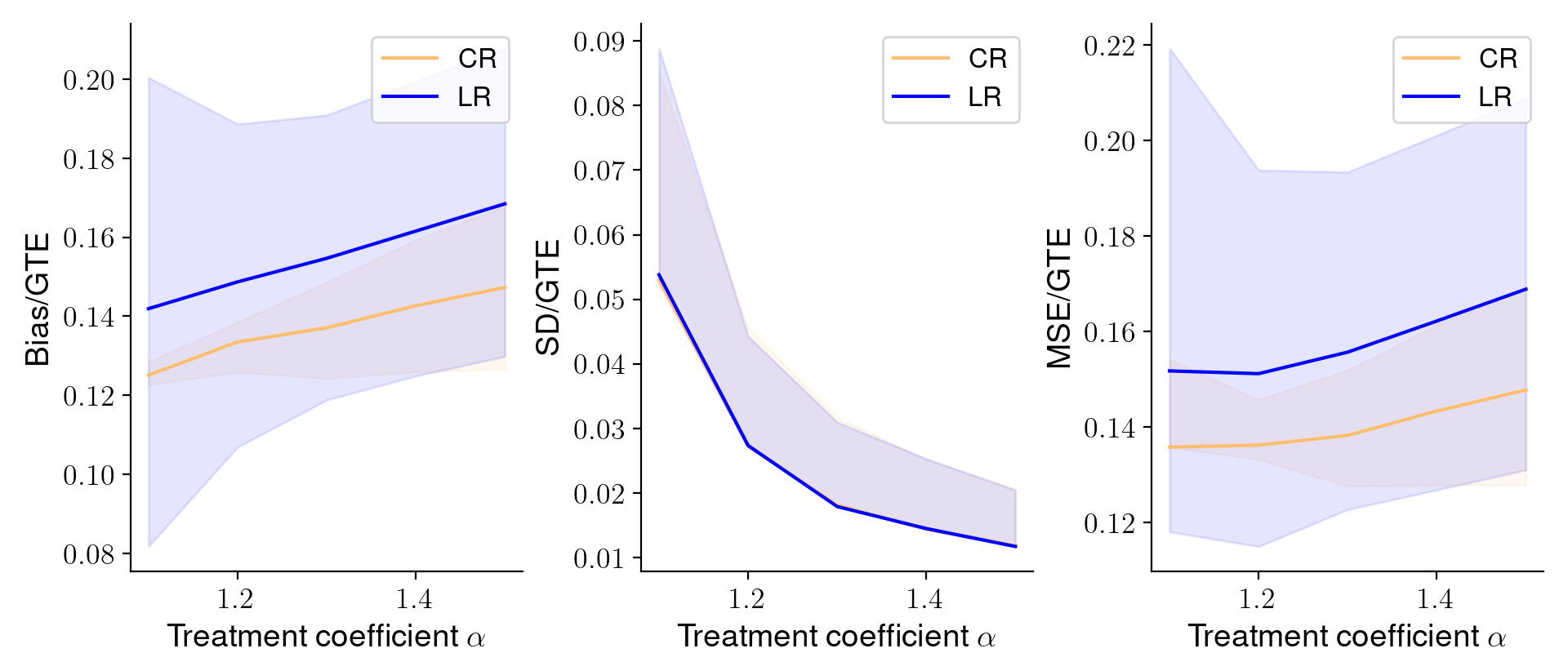}
    \caption{Balanced market ($M=N$). In this example, we consider $M=N=2^{20}\approx 1.0$ million homogeneous customers and listings, with $\conrate=N\conprob=0.253$, giving a global control booking probability of $20\%$.}
    \label{fig:vary_treat_balanced_demand}
\end{figure}

\begin{figure}[H]
    \centering
    \includegraphics[scale=0.5]{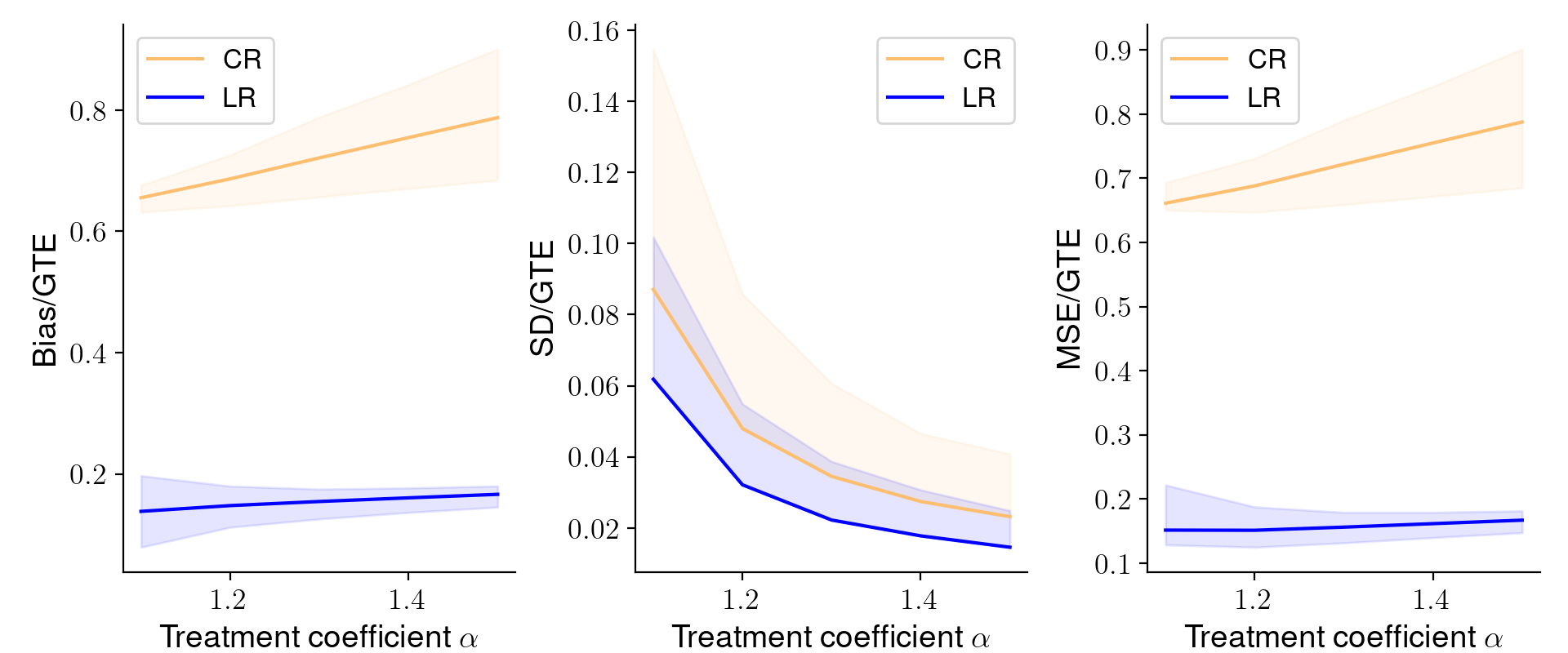}
    \caption{Supply constrained market. In this example, we consider $M=2^{20}\approx 1.0$ million homogeneous customers and $N=M/4=2^{18}\approx 267$ thousand homogeneous listings, with $\conrate=N\conprob=0.253$.}
    \label{fig:vary_treat_more_demand}
\end{figure}

We then consider heterogeneous markets. Figure \ref{fig:vary_treat_het_customers} shows simulations in a market where heterogeneous customer types and homogeneous listings. Figure \ref{fig:vary_treat_het_listings} shows simulations in a market with homogeneous customers and heterogeneous listings. 

\begin{figure}[H]
    \centering
    \includegraphics[scale=0.5]{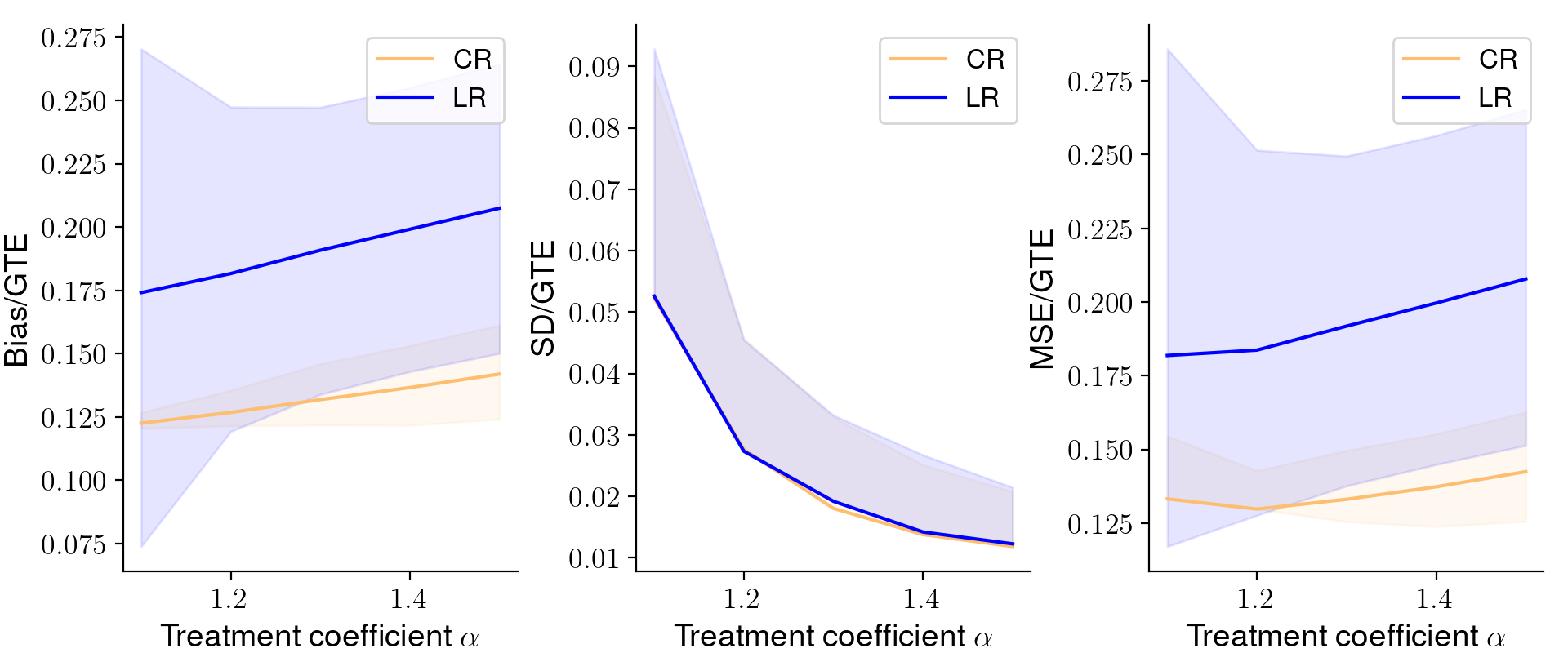}
    \caption{Heterogeneous customers. We consider a balanced market consisting of $M=N=2^{20}\approx 1$ million listings and customers, where the listings are homogeneous, and customers are split into two types $\gamma_1$ and $\gamma_2$. The customer type proportions are given by $\v{s} = (0.4, 0.6)$ and the consideration probabilities are defined with $\v{\conrate}^T = N\v{\conprob}^T = (0.101, 0.354)$.}
    \label{fig:vary_treat_het_customers}
\end{figure}

\begin{figure}[H]
    \centering
    \includegraphics[scale=0.5]{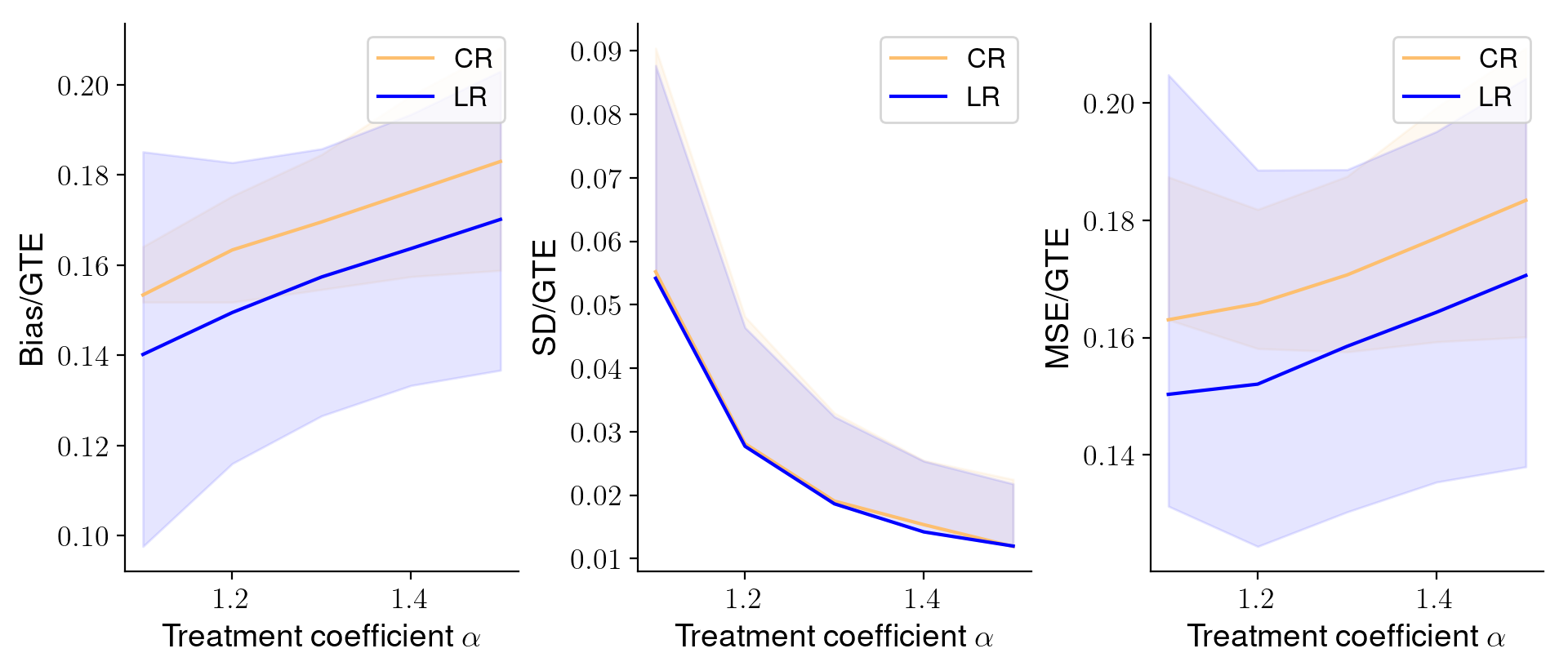}
    \caption{Heterogeneous listings. We consider a balanced market consisting of $M=N=2^{20}\approx 1$ million listings and customers, where the customers are homogeneous, and listing are split into two types $\theta_1$ and $\theta_2$. The listing type proportions are given by $\v{t} = (0.4, 0.6)$ and, as in the previous example, the consideration probabilities are defined with $\v{\conrate} = N\v{\conprob} = (0.101,\, 0.354)$.}
    \label{fig:vary_treat_het_listings}
\end{figure}

We then turn our attention to the dependence of bias, standard deviation, and $\MSE$ on the treatment allocation. Figures \ref{fig:vary_alloc_small_market}-\ref{fig:vary_alloc_large_market} show simulations in a balanced market with $M=N$. In a small market (Figure \ref{fig:vary_alloc_small_market}), the allocation that minimizes $\MSE$ is closer to the variance-optimizing allocation, whereas in a large market (Figure \ref{fig:vary_alloc_large_market}) the bias-optimizing allocation minimizes $\MSE$. In each of the three figures, the parameters are as follows: global control booking percentage of $20\%$ and global treatment booking percentage of $22\%$, meaning $\conrate=N\conprob=0.253$ and $\tildeconrate=N\tildeconprob=0.286$.

\begin{figure}[H]
    \centering
    \includegraphics[scale=0.5]{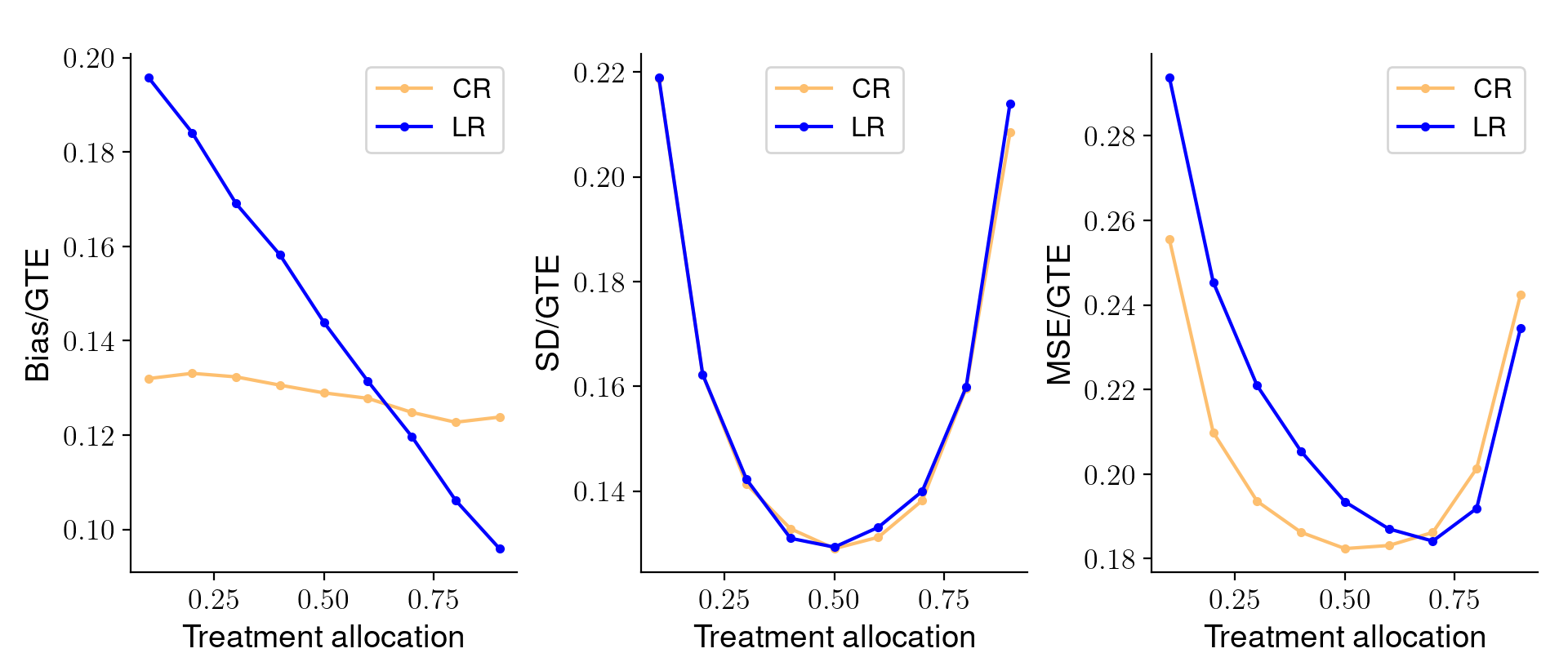}
    \caption{Small market ($M=N=10^5$)}
    \label{fig:vary_alloc_small_market}
\end{figure}

\begin{figure}[H]
    \centering
    \includegraphics[scale=0.5]{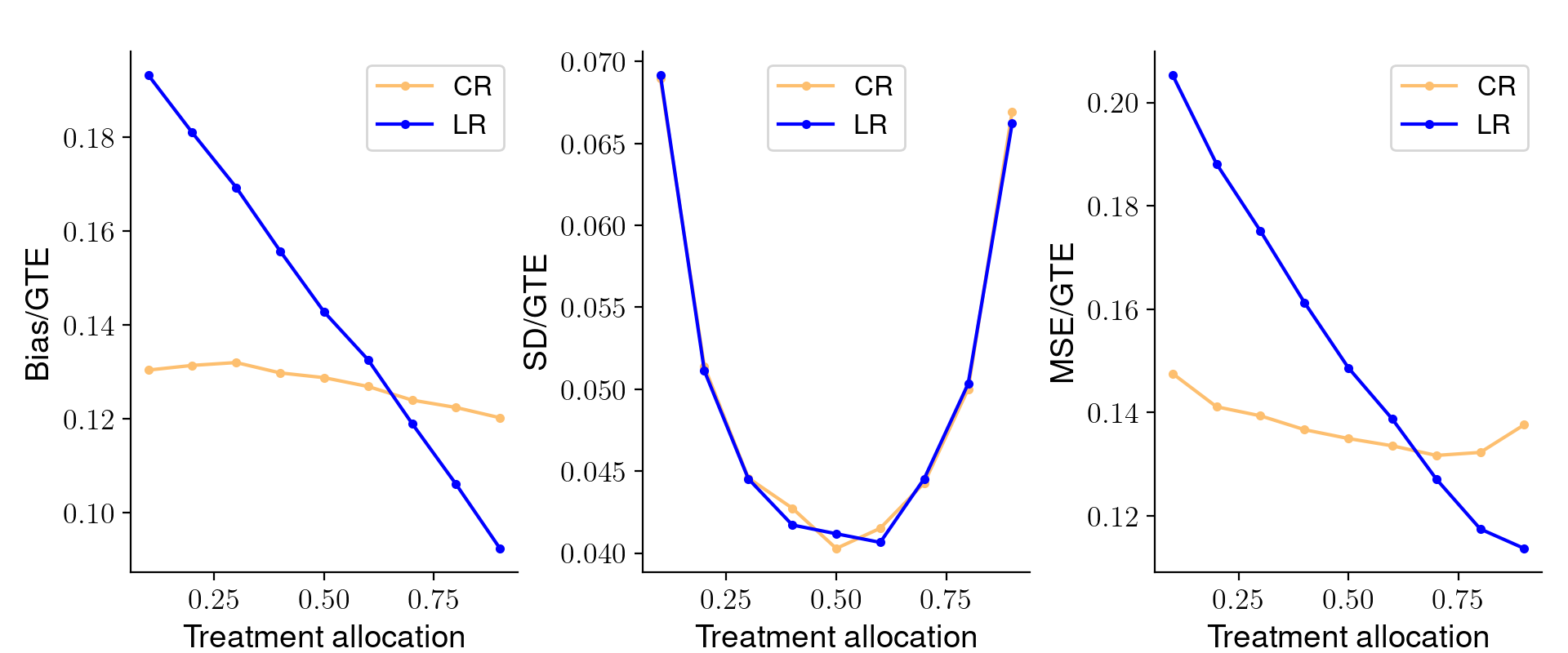}
    \caption{Medium market ($M=N=10^6$)}
    \label{fig:vary_alloc_med_market}
\end{figure}

\begin{figure}[H]
    \centering
    \includegraphics[scale=0.5]{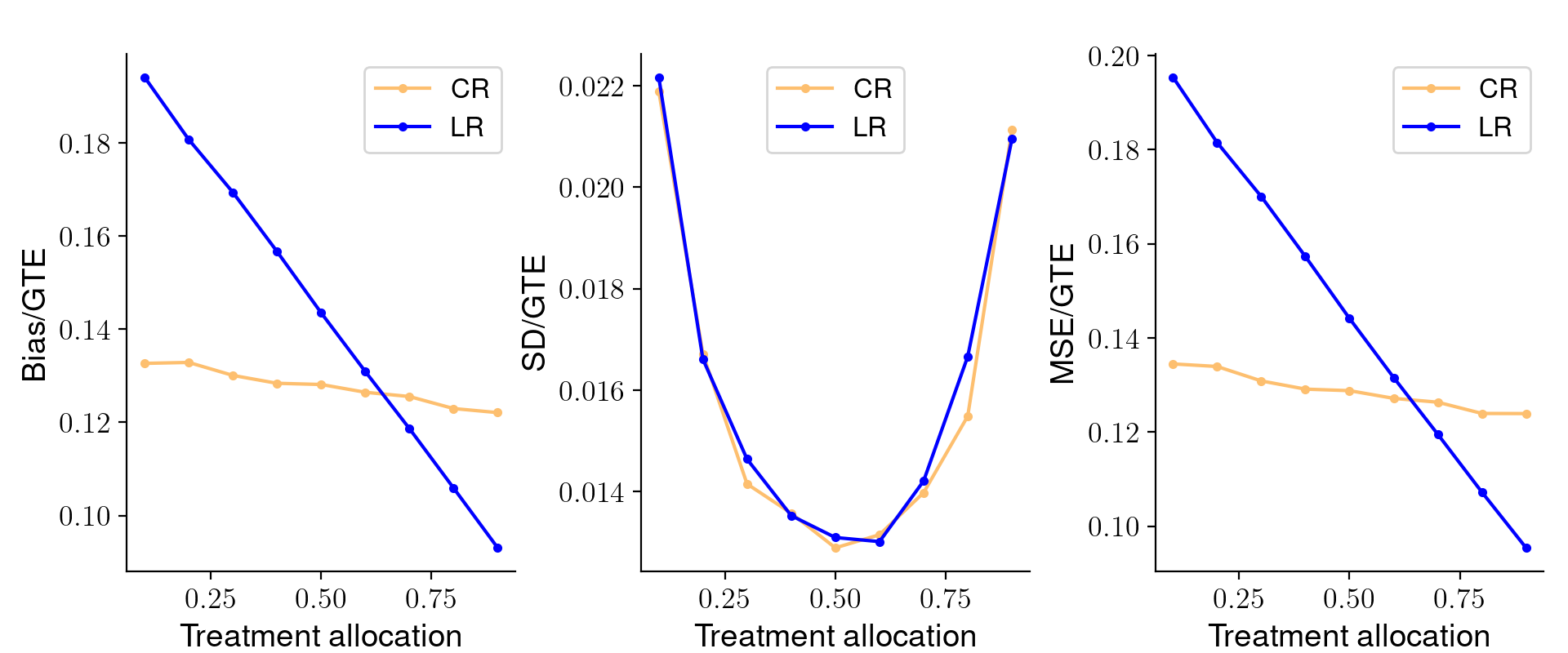}
    \caption{Large market ($M=N=10^7$)}
    \label{fig:vary_alloc_large_market}
\end{figure}

The plots below show how variance changes as we increase the treatment allocation, in a balanced market. The different curves represent different levels of $\GTE$. We find that for reasonable $\GTE$, the variance-minimizing allocation is roughly $0.5$. 

\begin{figure}[H]
    \centering
    \begin{subfigure}[b]{0.495\textwidth}
        \centering
        \includegraphics[width=\textwidth]{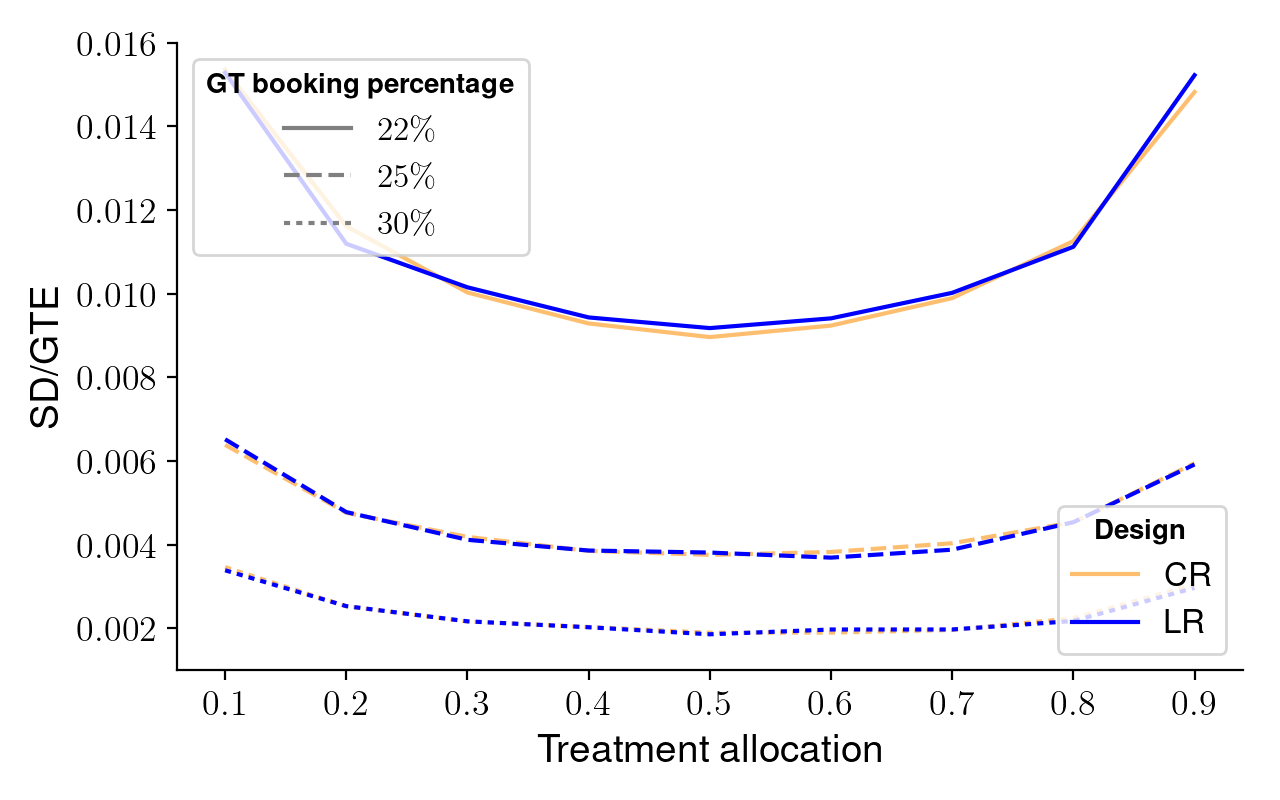}
        \caption{Simulation results}
        \label{fig:vary_alloc_vary_gte_sim}
    \end{subfigure}
    \hfill
    \begin{subfigure}[b]{0.495\textwidth}
        \centering
        \includegraphics[width=\textwidth]{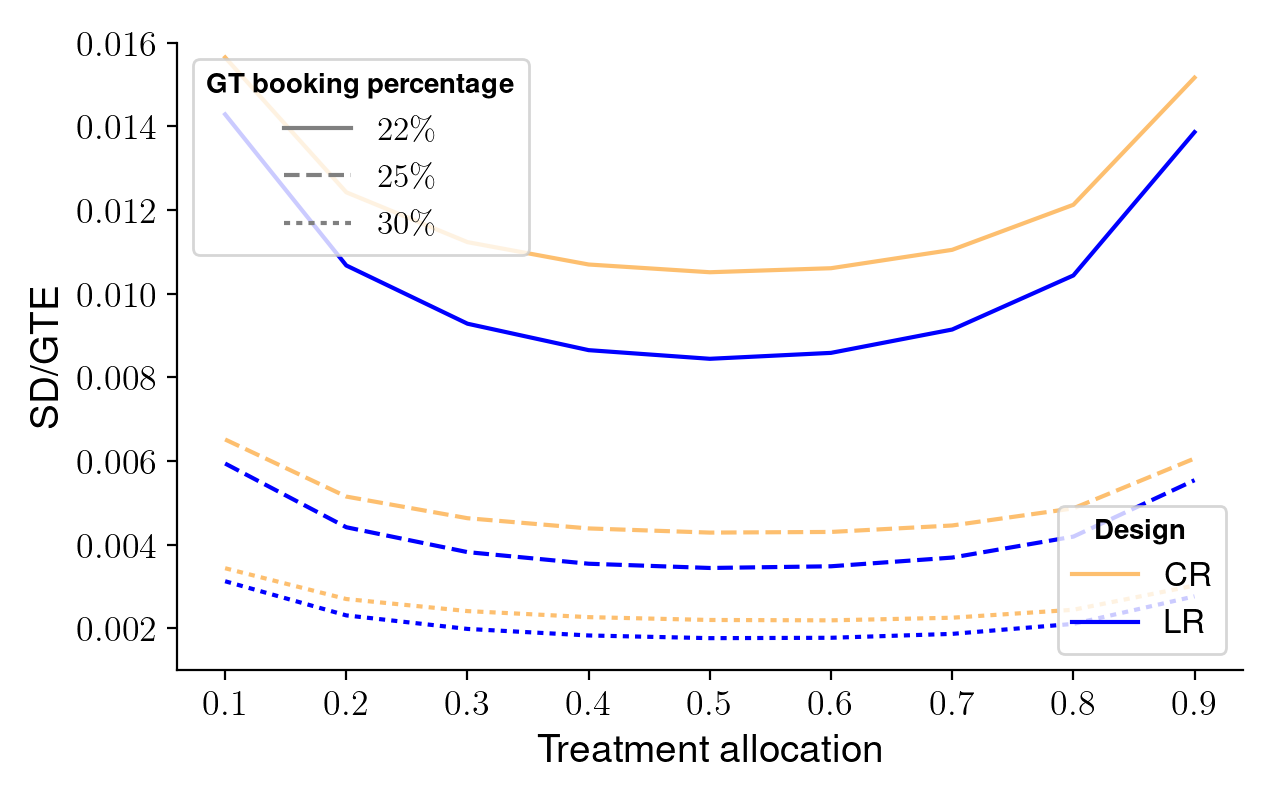}
        \caption{Numeric results}
        \label{fig:vary_alloc_vary_gte_num}
    \end{subfigure}
    \caption{Simulation and numeric results showing how standard deviation changes as treatment allocation proportion changes. Colors correspond to $\CR$ and $\LR$ designs and different curves indicate different levels of global treatment booking percentages. Left: Simulations in a homogeneous market of 20 million listings and customers with $20\%$ booking percentage in global control ($\conrate=N\conprob=0.253$). Right: Numeric evaluation of the asymptotic limit of scaled variance with $20\%$ booking rate in global control ($\conrate=0.253$), scaled by market size $N=M=20$ million to match the simulation scale.}
    \label{fig:vary_alloc_vary_gte}
\end{figure}
\section{Asymptotic results for variances of the $\LR$ and $\CR$ estimators}\label{sec_variance_appendix}

In this section, we will develop basic results to characterize the variances of the naive $\CR$ and $\LR$ estimators. The lemmas will immediately imply the variance scales at order $1/N$. Further, the results offer a closed form for the (constant) limits of the scaled variances, i.e., $N\Var\big(\gtehatlr(a_L)\big)$ and $N\Var\big(\gtehatcr(a_C)\big)$.

In the derivation below, we will assume for simplicity that the market is homogeneous; the techniques and results should generalize to heterogeneous settings as well, although the expressions become much more complicated.

We begin with expressing the variance of the naive $\LR$ estimator as
\begin{align}
    \Var\left(\gtehatlr(a_L)\right)
    &= \Var\left(N_1^{-1}\sum_{l:\;Z_l=1} Y_l - N_0^{-1}\sum_{l':\;Z_{l'}=0} Y_l'\right) \nonumber \\
    &= \frac{1}{N_1}\Var(Y_T) + \frac{1}{N_0}\Var(Y_C) \nonumber \\
    & \qquad + \frac{N_1-1}{N_1}\Cov(Y_T,Y_{T'}) + \frac{N_0-1}{N_0}\Cov(Y_{C},Y_{C'}) - 2\Cov(Y_T,Y_{C}),\label{Eqn_LR_var_expanded_Yl}
\end{align}
where $Y_l$ denotes the indicator random variable that listing $l$ is booked, and $Y_T$ (resp. $Y_C$) denotes the indicator that a generic treatment (resp. control) listing is booked, due to the symmetry we imposed on the listings.

\begin{lemma}[Limit of scaled variance of the $\LR$ estimator]\label{Lemma_limit_scaled_lr_var}
In homogeneous markets, as $N\to\infty$ we have
\begin{equation}\label{Eqn_limit_scaled_lr_var:Lemma_limit_scaled_lr_var}
    N\Var\big(\gtehatlr(a_L)\big) \to VT_\LR + VC_\LR + CV_\LR,
\end{equation}
where $VT_\LR, VC_\LR$, and $CV_\LR$ represent the total contribution from the variance of $Y_l$ for all treated listings $l$, the variance of $Y_l$ for all control listings $l$, and the covariance between all pairs different listings, given by
\[
    VT_\LR = \frac{1}{a_L} \exp\big(-\lambda \tildeconrate F((1-a_L)\conrate + a_L \tildeconrate)\big) \left(1-\exp\big(-\lambda \tildeconrate F((1-a_L)\conrate + a_L \tildeconrate)\big)\right),
\]
\[
    VC_\LR = \frac{1}{1-a_L} \exp\big(-\lambda \conrate F((1-a_L)\conrate + a_L \tildeconrate)\big) \left(1-\exp\big(-\lambda \conrate F((1-a_L)\conrate + a_L \tildeconrate)\big)\right),
\]
\[
    CV_\LR = -\lambda \left( \tildeconrate \exp\big(-\lambda \tildeconrate F((1-a_L)\conrate + a_L \tildeconrate)\big)- \conrate \exp\big(-\lambda \tildeconrate F((1-a_L)\conrate + a_L \tildeconrate)\big) \right)^2.
\]
\end{lemma}
\begin{remark}
Notice that the contribution from covariance terms is always non-positive in the limit.
\end{remark}
\begin{proof}
For a given value of $N$, $\E[Y_l]$ for a listing $l$ is a Bernoulli random variable with expectation of the following form
\[
    \E [Y_l] = \begin{cases}1 - (1-\appprob_0)^{M}, &\text{ if } Z_l=0, \\ 1 - (1-\appprob_1)^M, &\text{ if } Z_l=1,\end{cases}
\]
where $\appprob_0$ (resp. $\appprob_1$) denotes the probability that a customer applies to a certain control (resp. treatment) listing. By Lemma~\ref{Prop_limit_app_and_book_prob} applied with listing type space $\{0,1\}$ (indicating treatment condition), we have as $N\to\infty$
\[
    N\appprob_0 \to \conrate F((1-a_L)\conrate + a_L \tildeconrate),
\]
\[
    N\appprob_1 \to \tildeconrate F((1-a_L)\conrate + a_L \tildeconrate).
\]
From the convergence $\lim_{N\to\infty}(1-a/N)^N = e^{-a}$, we immediately obtain the limits for $N\Var(Y_T)\to a_L VT_\LR$ and $N\Var(Y_C)\to (1-a_L) VC_\LR$ in \eqref{Eqn_LR_var_expanded_Yl}.

Similarly, we express the expectation $\E[Y_lY_{l'}]$ for two different listings $l$ and $l'$ in closed form, and from the convergence
\[
    \lim_{N\to\infty} N\left( \Big(1 - \frac{a}{N} + \frac{b}{N^2}\Big)^{N} - \Big(1 - \frac{a}{N}\Big)^{N} \right) = be^{-a},
\]
we recover the combined contribution to $N\Var\big(\gtehatlr(a_L)\big)$ of the covariance terms in \eqref{Eqn_LR_var_expanded_Yl}.
\end{proof}

Now let's switch to the variance of $\gtehatcr$. While it is tempting to apply the similar analysis by breaking the estimator into indicator random variables for each control and treatment customer, it appears nontrivial to handle the covariance caused when listings pick one application each to accept. We instead define indicator random variables $Y_{l,0}$ and $Y_{l,1}$ for each listing $l$ to indicate whether $l$ is matched to a control or treatment customer. In this notation, we may express the naive $\CR$ estimator as
\[
    \gtehatcr(a_C) = \frac{M}{N} \sum_{l} \left(\frac{Y_{l,1}}{M_1} - \frac{Y_{l,0}}{M_0}\right),
\]
and hence its variance as
\begin{align}
    \Var\left(\gtehatcr(a_C)\right) &= \frac{M^2}{N^2}\left(\frac{N}{M_1^2}\Var(Y_{l,1} ) + \frac{N}{M_0^2}\Var(Y_{l,0})\right. \nonumber \\
    &\qquad + \frac{N^2-N}{M_1^2} \Cov(Y_{l,1}, Y_{l',1}) + \frac{N^2-N}{m_0^2} \Cov(Y_{l,0}, Y_{l',0}) \nonumber \\
    &\qquad \left. - \frac{2N^2-2N}{M_0M_1} \Cov(Y_{l,1}, Y_{l',0}) - \frac{2N}{M_0M_1} \Cov(Y_{l,1}, Y_{l,0})\right). \label{Eqn_Var_CR_hat_full_expansion}
\end{align}
The limit for the scaled variance of the $\CR$ estimator turns out rather lengthy, and a term-by-term breakdown is presented in the following lemma.

\begin{lemma}[Limit of scaled variance of the $\CR$ estimator]\label{Lemma_limit_scaled_cr_var}
In homogeneous markets, as $N\to\infty$, we have
\[
    N\Var\left(\gtehatcr(a_C)\right) \to VT_\CR + VC_\CR + CVTT_\CR + CVCC_\CR + CVTC_\CR.
\]
Recall that we let $\apprate$ and $\tildeapprate$ denote the limit of scaled application probability under global control and global treatment. with the terms representing, respectively, the total contribution from the variances of $Y_{l,1}$, the variances of $Y_{l,0}$, the covariances between $Y_{l,1}$ and $Y_{l',1}$ and between $Y_{l,0}$ and $Y_{l',0}$ for pairs of different listings, and finally the covariances between $Y_{l_1,1}$ and $Y_{l_2,0}$ where $l_1,l_2$ may be the same or different listings. With the notation of average application rate $\bar{\apprate}=(1-a_C)\apprate + a_C\tildeapprate$, these terms can be expressed as follows
\begin{equation}
    VT_\CR = \frac{1}{a_C} \frac{ \tildeapprate}{\bar{\apprate}}(1-\exp(-\lambda\bar{\apprate})) \left(1 - \frac{ a_C\tildeapprate}{\bar{\apprate}}(1-\exp(-\lambda\bar{\apprate}))\right),
\end{equation}
\begin{equation}
    VC_\CR = \frac{1}{1-a_C} \frac{\apprate}{\bar{\apprate}}(1-\exp(-\lambda\bar{\apprate})) \left(1 - \frac{(1-a_C)\apprate}{\bar{\apprate}}(1-\exp(-\lambda\bar{\apprate}))\right),
\end{equation}
\begin{align}
    CVTT_\CR &= \frac{1}{a_C} \left( -\frac{(1-a_C)\apprate \tildeapprate}{\lambda\bar{\apprate}^3}(1-\exp(-\lambda\bar{\apprate}))^2 - \frac{a_C\tildeapprate^2}{\bar{\apprate}} \lambda\exp(-2\lambda\bar{\apprate}) \right.  \nonumber\\
    &\qquad\qquad + (1-a_C) a_C\lambda \apprate(1-\apprate) \cdot
    \left(\frac{- \tildeapprate(1-\exp(-\lambda\bar{\apprate})) + \lambda \tildeapprate \bar{\apprate} \exp(-\lambda\bar{\apprate})}{\lambda\bar{\apprate}^2}\right)^2  \nonumber\\
    &\qquad\qquad \left. + \lambda  \tildeapprate(1- \tildeapprate) \cdot
    \left(\frac{(1-a_C)\apprate(1-\exp(-\lambda\bar{\apprate})) + \lambda a_C\tildeapprate \bar{\apprate} \exp(-\lambda\bar{\apprate})}{\lambda\bar{\apprate}^2}\right)^2 \right),
\end{align}
\begin{align}
    CVCC_\CR &= \frac{1}{1-a_C} \left( -\frac{a_C\apprate \tildeapprate}{\lambda\bar{\apprate}^3}(1-\exp(-\lambda\bar{\apprate}))^2 - \frac{(1-a_C)\apprate^2}{\bar{\apprate}} \lambda\exp(-2\lambda\bar{\apprate}) \right.  \nonumber\\
    &\qquad\qquad + \lambda \apprate(1-\apprate) \cdot
    \left(\frac{ a_C\tildeapprate(1-\exp(-\lambda\bar{\apprate})) + \lambda(1-a_C)\apprate \bar{\apprate} \exp(-\lambda\bar{\apprate})}{\lambda\bar{\apprate}^2}\right)^2  \nonumber\\
    &\qquad\qquad \left. + a_C(1-a_C)\lambda  \tildeapprate(1- \tildeapprate) \cdot
    \left(\frac{-\apprate(1-\exp(-\lambda\bar{\apprate})) + \lambda\apprate \bar{\apprate} \exp(-\lambda\bar{\apprate})}{\lambda\bar{\apprate}^2}\right)^2 \right),
\end{align}
\begin{align}
    CVTC_\CR &= - 2\left( -\frac{\apprate \tildeapprate}{\bar{\apprate}} \lambda\exp(-2\lambda\bar{\apprate}) \right.  \nonumber\\
    &\qquad\qquad + \lambda \apprate(1-\apprate) \cdot
    \frac{- \tildeapprate(1-\exp(-\lambda\bar{\apprate})) + \lambda \tildeapprate \bar{\apprate} \exp(-\lambda\bar{\apprate})}{\lambda\bar{\apprate}^2} \nonumber\\
    &\qquad\qquad\qquad \cdot\frac{ a_C\tildeapprate(1-\exp(-\lambda\bar{\apprate})) + \lambda(1-a_C)\apprate \bar{\apprate} \exp(-\lambda\bar{\apprate})}{\lambda\bar{\apprate}^2}  \nonumber\\
    &\qquad\qquad + \lambda  \tildeapprate(1- \tildeapprate) \cdot \frac{(1-a_C)\apprate(1-\exp(-\lambda\bar{\apprate})) + \lambda a_C\tildeapprate \bar{\apprate} \exp(-\lambda\bar{\apprate})}{\lambda\bar{\apprate}^2} \nonumber\\
    &\qquad\qquad\qquad \left. \cdot
    \frac{-\apprate(1-\exp(-\lambda\bar{\apprate})) + \lambda\apprate \bar{\apprate} \exp(-\lambda\bar{\apprate})}{\lambda\bar{\apprate}^2} \right) \nonumber\\ 
    &\qquad +  \frac{2\apprate \tildeapprate}{\bar{\apprate}^2}(1-\exp(-\lambda\bar{\apprate}))^2 .
\end{align}

\end{lemma}

The full proof is, not surprisingly, tedious and not particularly insightful; we omit it in the interest of space, but briefly comment on the following derivation for the limits of the variance terms.

Since $\Var(Y_{l,0}) = \mathbb{E}[Y_{l,0}] - \left(\mathbb{E}[Y_{l,0}]\right)^2$, it suffices for us to show that
\begin{equation}\label{Eqn_prf_lemma_CR_var_part_I_var_Y_Cl}
    \E[Y_{l,0}] \to \frac{(1-a_C)\apprate}{\bar{\apprate}}(1-\exp(-\lambda\bar{\apprate})).
\end{equation}
We will condition on the outcome of the application step. Let $A_0$ and $A_1$ denote the total number of applications sent from control and treatment customers, respectively, and let $B$ denote the number of listings receiving at least one applications. We know that for any listing $l$,
\[
    \mathbb{E}[Y_{l,0} | A_0, A_1, B] = \frac{B}{N}\cdot \frac{A_0}{A_0+A_1},
\]
and
\[
    \mathbb{E}(Y_{l,1} | A_0, A_1, B) = \frac{B}{N}\cdot \frac{A_1}{A_0+A_1}.
\]
By the law of total expectation,
\[
    \mathbb{E}[Y_{l,0}] = \mathbb{E}\left[\mathbb{E}\left[Y_{l,0} \middle| A_0, A_1\right]\right].
\]
The inner conditional expectation is equal to the conditional probability that listing $\ell$ is matched to a control customer given that control and treatment customers make a total of $A_0$ and $A_1$ applications, respectively. Hence,
\[
    \mathbb{E}\left[Y_{l,0} \middle| A_0, A_1\right] = \mathbb{P}\left(Y_{l,0}=1 \middle| A_0, A_1,Y_\ell = 1\right) \mathbb{P}\left(Y_{\ell} \middle| A_0, A_1\right) = \frac{A_0}{A_0+A_1}\left(1-\left(1-N^{-1}\right)^{A_0+A_1}\right).
\]
As a result of the law of large numbers and that $A_0/M\overset{p}{\to}(1-a_C)\apprate$ and $A_1/M\overset{p}{\to} a_C\tildeapprate$. Therefore,
\[
    \mathbb{E}\left[Y_{l,0} \middle| A_0, A_1\right] \overset{p}{\to} \frac{(1-a_C)\apprate}{(1-a_C)\apprate+ a_C\tildeapprate}(1-\exp(-\lambda((1-a_C)\apprate+ a_C\tildeapprate))),
\]
and the convergence is also in expectation by bounded convergence theorem. This proves \eqref{Eqn_prf_lemma_CR_var_part_I_var_Y_Cl} and therefore proves the convergence of $\Var(Y_{l,1})\to VT_\CR$. The results for the other terms in Lemma~\ref{Lemma_limit_scaled_cr_var} follows the similar analysis by conditioning on the state of the applications.

\section{Numeric results for variances of the $\CR$ and $\LR$ estimators}

To complement our asymptotic results about variances of the estimators, we offer numeric evidence for the approximate variance-optimality with a treatment allocation ratio of $0.5$ in both $\CR$ and $\LR$ designs under most practical situations. We are interested in the asymptotic approximation ratios using $0.5$ allocation compared with the variance-optimal allocation in the mean field limit, namely
\[
\beta_\CR \vcentcolon= \lim_{N\to\infty} \frac{\Var\big(\gtehatcr(0.5)\big)}{\Var\big(\gtehatcr(a_C^*)\big)}
\qquad\text{ and }\qquad
\beta_\LR \vcentcolon= \lim_{N\to\infty} \frac{\Var\big(\gtehatlr(0.5)\big)}{\Var\big(\gtehatlr(a_L^*)\big)},
\]
where $a_C^* = \argmin\; \Var\big(\gtehatcr(a_C)\big)$ and $a_L^* = \argmin\; \Var\big(\gtehatlr(a_L)\big)$ are the corresponding variance-minimizing allocation ratios. 
These quantities are numerically computed using the formulas derived in the previous section. Figure~\ref{fig:cr_var_approx_ratio} and \ref{fig:lr_var_approx_ratio} show the approximation ratios in $\CR$ and $\LR$ designs, respectively. Note that the color scales are different in the two figures: namely, in Figure~\ref{fig:cr_var_approx_ratio} the largest ratio across the entire range of market balance and treatment effect is $1.004$. The approximation ratio only deteriorates noticeably for $\LR$ experiments when the treatment effect is large and when the market is largely supply constrained (i.e. $\lambda$ is large), a situation that we believe is rather uncommon in practical settings.

\begin{figure}[H]
    \centering
    \includegraphics[scale=0.55]{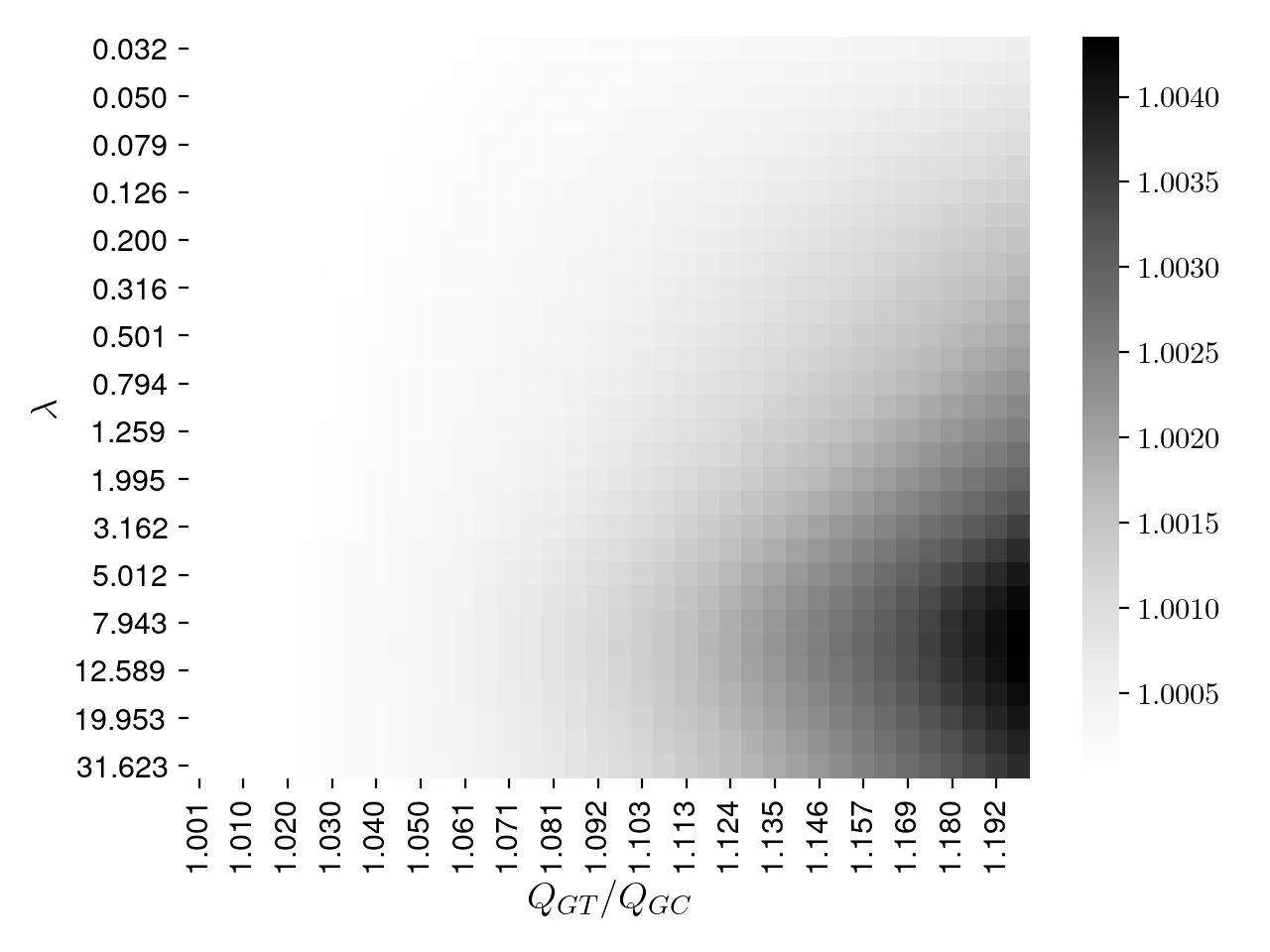}
    \caption{Approximation ratio $\beta_\CR$ in asymptotic variance of the allocation $a_C=0.5$ in $\CR$ experiments.}
    \label{fig:cr_var_approx_ratio}
\end{figure}

\begin{figure}[H]
    \centering
    \includegraphics[scale=0.55]{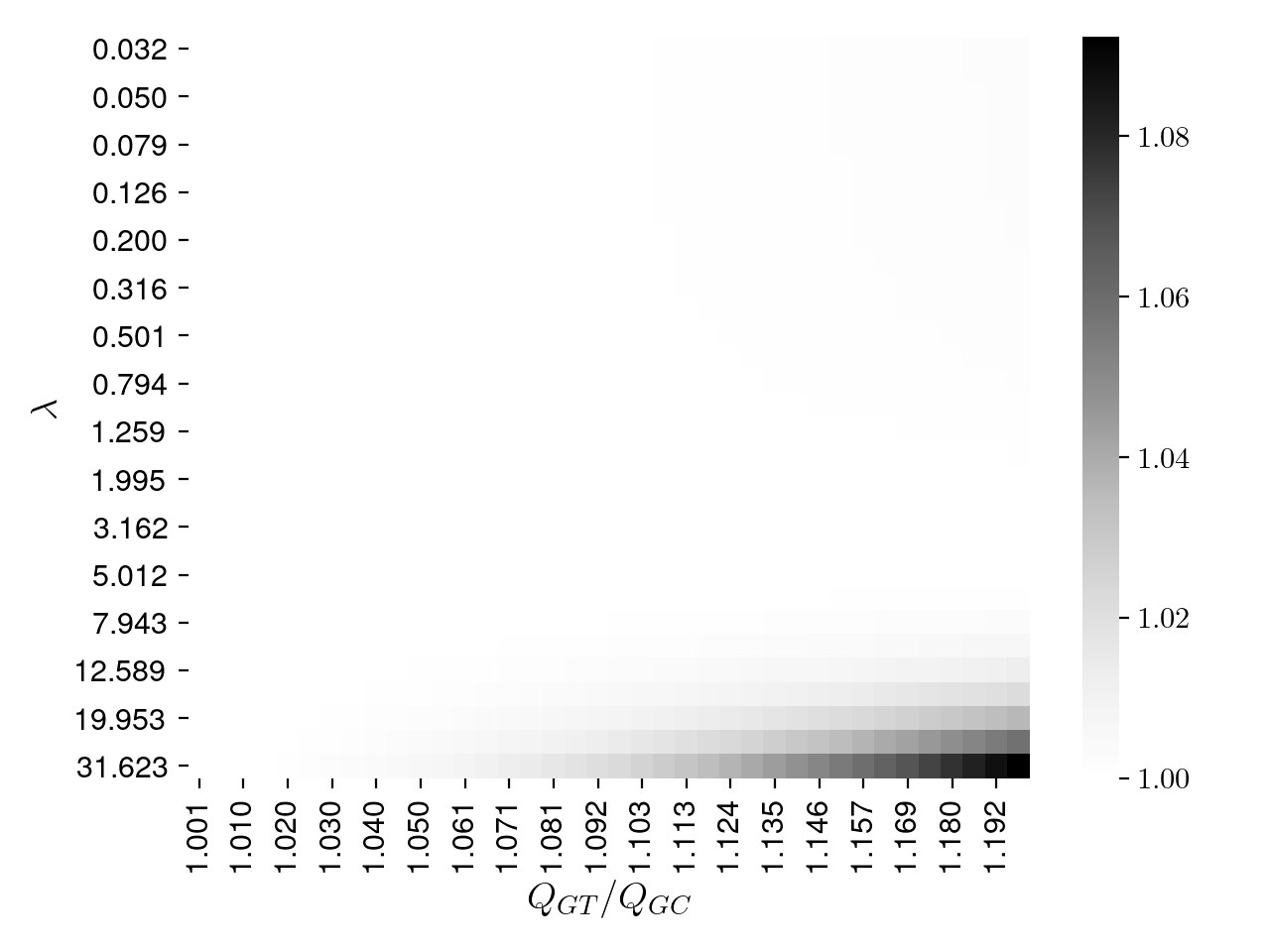}
    \caption{Approximation ratio $\beta_\LR$ in asymptotic variance of the allocation $a_L=0.5$ in $\LR$ experiments.}
    \label{fig:lr_var_approx_ratio}
\end{figure}

\end{document}